\documentclass{article}
\usepackage[pdftex]{graphicx}
\graphicspath{{Figs/}}
\usepackage{comment}
\usepackage[]{algorithm2e}
\usepackage[a4paper]{geometry}  
\usepackage{amsmath,amsfonts,amssymb,amscd,amsthm,xspace}
\usepackage{amsmath,amsfonts,amssymb,amscd,amsthm,xspace}
\usepackage{algorithmic}
\usepackage{fixltx2e}

\usepackage{epsfig}
\usepackage{enumitem}
\usepackage{bbm}
\usepackage{dblfloatfix}
\usepackage{bm}
\usepackage{color}
\usepackage[utf8]{inputenc} 
\usepackage[TS1,T1]{fontenc}
\usepackage{newtxtext}
\usepackage{mdframed}
\usepackage{graphicx,epstopdf}

\newtheorem{theorem}{\textit{Theorem}}
\numberwithin{theorem}{section}
\newtheorem{lemma}{\textit{Lemma}}
\newtheorem{corollary}{\textit{Corollary}}
\numberwithin{corollary}{section}

\theoremstyle{definition}
\newtheorem{definition}{\textit{Definition}}
\numberwithin{definition}{section}
\newtheorem{example}{\textit{Example}}
\numberwithin{example}{section}

\newtheorem{remark}{\textit{Remark}}

\newcommand{\nn}{n}
\newcommand{\ns}{m}
\newcommand{\nk}{K}

\newcommand{\poet}{{\varepsilon}_{\mathrm{T}}}
\newcommand{\poe}{{\varepsilon}_{\mathrm{P}}}

\DeclareMathOperator*{\argmin}{arg\,min}
\begin{document}
%
\title{Learning Graphs from Linear Measurements:
Fundamental Trade-offs and Applications
}
%
%
%

\author{Tongxin~Li, Lucien~Werner and Steven H. Low
\thanks{Li, Werner and Low are with the Computing + Mathematical Sciences Department, California Institute of Technology, Pasadena, CA 91125 USA (e-mails: \{tongxin, lwerner, slow\}@caltech.edu)}}

\maketitle

\begin{abstract}
We consider a specific graph learning task: reconstructing a symmetric matrix that represents an underlying graph using linear measurements. We present a sparsity characterization for distributions of random graphs (that are allowed to contain \textit{high-degree} nodes), based on which we study fundamental trade-offs between the number of measurements, the complexity of the graph class, and the probability of error. We first derive a necessary condition on the number of measurements. Then, by considering a three-stage recovery scheme, we give a sufficient condition for recovery. Furthermore, assuming the measurements are Gaussian IID, we prove upper and lower bounds on the (worst-case) sample complexity for both noisy and noiseless recovery. In the special cases of the uniform distribution on trees with $\nn$ nodes and the Erd\H{o}s-R\'{e}nyi $(\nn,p)$ class, the fundamental trade-offs are tight up to multiplicative factors with noiseless measurements. 
In addition, for practical applications, we design and implement a polynomial-time (in $\nn$) algorithm based on the three-stage recovery scheme. Experiments show that the heuristic algorithm outperforms basis pursuit on star graphs. We apply the heuristic algorithm to learn admittance matrices in electric grids. Simulations for several canonical graph classes and IEEE power system test cases demonstrate the effectiveness and robustness of the proposed algorithm for parameter reconstruction.
\end{abstract}


%
\maketitle

\section{Introduction}
\label{sec:int}
\subsection{Background}
\label{sec:int.bac}
S{ymmetric}
matrices are ubiquitous in graphical models with examples such as the $(0,1)$ adjacency matrix and the (generalized) Laplacian of an undirected graph. A major challenge in graph learning is inferring graph parameters embedded in those graph-based matrices from historical data or real-time measurements. In contrast to traditional statistical inference methods~\cite{chow1968approximating,chow1973consistency,tan2010learning}, model-based graph learning, such as physically-motivated models and graph signal processing (GSP)~\cite{dong2019learning}, takes advantage of additional data structures offered freely by nature. Among different measurement models for graph learning, linear models have been used and analyzed widely for different tasks, \textit{e.g.,} linear structural
equation models (SEMs)~\cite{ghoshal2018learning,ghoshal2017learning}, linear graph measurements~\cite{ahn2012analyzing}, generalized linear cascade models~\cite{pouget2015inferring}, \textit{etc}.

Despite extra efforts required on data collection, processing and storage, model-based graph learning often guarantees provable sample complexity, which is often significantly lower than the empirical number of measurements needed with traditional inference methods. In many problem settings, having computationally efficient algorithms with low sample complexity is important. One reason for this is that the graph parameters may change in a short time-scale, making sample complexity a vital metric to guarantee that the learning can be accomplished with limited measurements. 
Indeed many applications, such as real-time optimal power flow~\cite{momoh1999review,low2014convex,tang2017real}, real-time contingency analysis~\cite{mittal2011real} and frequency control~\cite{horta2015frequency} in power systems \textit{etc}., require data about the network that are time-varying.
For example, the generations or net loads  may change rapidly due to the proliferation of distributed energy resources. The topology and line parameters of the grid may be reconfigured to mitigate
cascading failure~\cite{guo2018failure}. Line switching has changed the traditional idea of a power network with a fixed topology, enabling power flow control by switching lines~\cite{zhao2013vulnerability}, \textit{etc.}
Hence analyzing fundamental limits of parameter reconstruction and designing graph algorithms that are efficient in both computational and sample complexity are important. 

The number of measurements needed for reconstructing a graph Laplacian can be affected by various system parameters, such as data quality (distribution), physical laws, and graph structures. In particular, existing recovery algorithms often assume the graph to be recovered is in a specific class, \textit{e.g.,} trees~\cite{chow1968approximating}, sparse graphs~\cite{chepuri2017learning}, graphs with no high-degree nodes~\cite{dasarathy2015sketching}, with notable exceptions such as~\cite{belilovsky2017learning}, which considers an empirical algorithm for topology identification. However, there is still a lack of understanding of sample complexity for learning general undirected graphs that may contain high-degree nodes, especially with measurements constrained naturally by a linear system.

In this work, we consider a general graph learning problem where the measurements and underlying matrix to be recovered can be represented as or approximated by a linear system. 
A \textit{graph matrix} $\mathbf{Y}(G)$ with respect to an underlying graph $G$, which may have \textit{high-degree} nodes (see Definition~\ref{def:graph_matrix}) is defined as an $\nn\times\nn$ symmetric matrix with each nonzero $(i,j)$-th entry corresponding to an edge connecting node $i$ and node $j$ where $\nn\in\mathbbm{N}_+$ is the number of nodes of the underlying \textit{undirected} graph. 
The diagonal entries can be arbitrary. The measurements are summarized as two $\ns\times \nn$ ($1\leq \ns\leq \nn$) real or complex matrices $\mathbf{A}$ and $\mathbf{B}$ satisfying
\begin{align}
\label{eq:linear}
\mathbf{A} = \mathbf{B}\mathbf{Y}(G)+\mathbf{Z}
\end{align}
where $\mathbf{Z}$ denotes additive noise.

We focus on the following problems:
\begin{itemize}
	\item  {{\textit{Fundamental Trade-offs}.}} {What is the \textit{minimum number} $m$ of linear measurements required for reconstructing the \textit{symmetric} matrix $\mathbf{Y}(G)$?}
    {Is there an algorithm \textit{asymptotically achieving} recovery with the minimum number of measurements?} As a special case, can we characterize the sample complexity when the measurements are Gaussian IID\footnote{This means the entries of the matrix $\mathbf{B}$ are IID normally distributed.}? 
    \item {{\textit{Applications to Electrical Grids}.}} 
    Do the theoretical guarantees on sample complexity result in a practical algorithm (in terms of both sample and computational complexity) for recovering electric grid topology and parameters?
\end{itemize}

Some comments about the above model and the results in this paper are as follows.

\begin{remark}
\label{remark:vec}
It has been noted that vectorization and standard compressed sensing techniques do not lead to straightforward results (see~\cite{dasarathy2015sketching} for detailed arguments about a similar linear system).
This issue is discussed extensively in Section~\ref{sec:lit_cs}. 
\end{remark}

\begin{remark}
The results in this paper do not assume low-degree nodes as most of existing results do, with notable exceptions such as~\cite{belilovsky2017learning} which gives empirical and data-based subroutines for topology identification.
\end{remark}

\subsection{Related Work}

\subsubsection{Graph Learning}
Algorithms for learning sparse graphical model structures have a rich tradition in the literature. For general  Markov random fields (MRFs), learning the underlying graph structures is known to be NP-hard~\cite{bogdanov2008complexity}. However, in the case when the underlying graph is a tree, the classical Chow-Liu algorithm~\cite{chow1968approximating} offers an efficient approach to structure estimation.
Recent results contribute to an extensive understanding of the Chow-Liu algorithm. The authors in~\cite{tan2010learning} analyzed the error exponent and showed experimental results for chain graphs and star graphs. For pairwise binary MRFs with bounded maximum degree,~\cite{santhanam2012information} provides sufficient conditions for correct graph selection. Similar achievability results for Ising models are in \cite{anandkumar2010high}. Model-based graph learning has been emerging recently and assuming the measurements form linear SEMs, the authors in~\cite{ghoshal2018learning,ghoshal2017learning} showed theoretical guarantees of the sample complexity for learning a directed acyclic graph (DAG) structure, under mild conditions on the class of graphs.

For converse, information-theoretic tools have been widely applied to derive fundamental limits for learning graph structures. For a Markov random
field with bounded maximum degree, necessary conditions on the number of samples for estimating the underlying graph structure were derived in~\cite{santhanam2012information} using Fano's inequality (see \cite{fano1961transmission}). For Ising models,~\cite{anandkumar2012high} combines Fano's inequality with the idea of \textit{typicality} to derive weak and strong converse.  Similar techniques have also been applied to Gaussian graphical models~\cite{anandkumar2011high} and  Bayesian networks~\cite{ghoshal2017information}. Fundamental limits for noisy compressed sensing  have been extensively studied in~\cite{aeron2010information} under an information-theoretic framework.

\subsubsection{System Identification in Power Systems}
Graph learning has been widely used in electric grids applications, such as state estimation~\cite{cavraro2017voltage,chen2015quickest} and topology identification~\cite{sharon2012topology,deka2017structure}. Most of the literature focuses on topology identification or change detection, but there is less work on joint topology and parameter reconstruction, with notable exceptions of~\cite{li2013blind,yuan2016inverse,yu2017patopa,park2018exact}. However, the linear system proposed in~\cite{yuan2016inverse} does not leverage the sparsity of the graph\footnote{With respect to sparsity, we consider not only graphs with bounded degrees, but a broader class of graphs which may contain high-degree nodes. Definition~\ref{def:sparse} gives a comprehensive characterization of sparsity.}. Thus, in the worst case, the matrix $\mathbf{B}$ needs to have full column rank, implying that $\ns=\Omega(\nn)$ measurements are necessary for recovery.

Moreover, there is little exploration on the fundamental performance limits (estimation error and sample complexity) on topology and parameter reconstruction of power networks, with the exception of~\cite{zhu2012sparse} where a sparsity condition was given for exact recovery of outage lines. Based on single-type measurements (either current or voltage), correlation analysis has been applied for topology identification~\cite{tan2011sample,liao2015distribution,bolognani2013identification}.  Approximating the measurements as normal distributed random variables, the authors of~\cite{sharon2012topology} proposed an approach for topology identification with limited measurements. A graphical learning-based approach can be found in~\cite{deka2016estimating}. Recently, data-driven methods were studied for parameter estimation~\cite{yu2017patopa}.
In~\cite{yuan2016inverse}, a similar linear system as~(\ref{2.3}) was used combined with regression to recover the symmetric graph parameters (which is the admittance matrix in the power network). 

\subsubsection{Compressed Sensing and Sketching}
\label{sec:lit_cs}
It is well known that compressed sensing~(\cite{candes2005error,rudelson2008sparse}) techniques allow for recovery of a sparse matrix with a limited number of measurements in various applications such as medical imaging~\cite{lustig2008compressed}, wireless communication~\cite{li2012compressed}, channel estimation~\cite{berger2010application} and circuit design ~\cite{li2016fundamental},~\textit{etc}. For electricity grids, in~\cite{babakmehr2016compressive}, based on these techniques, experimental results have been given for topology recovery. However, nodal admittance matrices (generalized Laplacians) for power systems have two properties for which there are gaps in the sparse recovery literature: 1) the presence of high-degree nodes in a graph (corresponding to dense columns in its Laplacian) and 2) symmetry. 

Consider a vectorization of system (\ref{eq:linear}) using tensor product notation, with $\mathbf{a}:=\mathrm{vec}(\mathbf{A})$ and  $\mathbf{y}(G):=\mathrm{vec}(\mathbf{Y}(G))$. Then linear system (\ref{eq:linear}) is equivalent to
$
  \mathbf{a} = (\mathbf{I}\otimes\mathbf{B})\mathbf{y}(G) 
$
where $\mathrm{vec}(\cdot)$ produces a column vector by stacking the columns of the input matrix and $\mathbf{I}\otimes\mathbf{B}$ is the Kronecker product of an identity matrix $\mathbf{I}\in\mathbbm{R}^{\nn\times\nn}$ and $\mathbf{B}$. With the sensing matrix being a Kronecker product of two matrices, traditional compressed sensing analysis works for the case when $\mathbf{y}$ contains only $\mu=\Theta(1)$ non-zeros~\cite{duarte2011kronecker}. For instance, the authors of~\cite{jokar2009sparse} showed that the restricted isometry constant (see Section~\ref{sec:def_of_rip} for the definition), $\delta_{\mu}(\mathbf{I}\otimes\mathbf{B})$ is bounded from above by $\delta_{\mu}(\mathbf{B})$, the restricted isometry constant of $\mathbf{B}$. However, if a column (or row) of $\mathbf{Y(G)}$ is dense, classical restricted isometry-based approach cannot be applied straightforwardly. 

Another way of viewing it is that vectorizing $\mathbf{A}$ and $\mathbf{Y}(G)$ and constructing a sensing matrix $\mathbf{I}\otimes\mathbf{B}$ is equivalent to recovering each of the column  (or row) of $\mathbf{Y}(G)$ separately from $A_j = \mathbf{B} Y_j(G)$ for $j=1,\ldots,\nn$ where $A_j$'s and $Y_j(G)$'s are columns of $\mathbf{A}$ and $\mathbf{Y}(G)$. For a general ``sparse" graph $G$, such as a star graph, some of the columns (or rows) of the graph matrix $\mathbf{Y}(G)$ may be dense vectors consisting of many non-zeros. The results in~\cite{jokar2009sparse,duarte2011kronecker} give no guarantee for the recovery of the dense columns of $\mathbf{Y}(G)$ (correspondingly, the high-degree nodes in $G$), and thus they cannot be applied directly to the analysis of sample complexity. This statement is further validated in our experimental results shown in Figure~\ref{fig:ieee30} and Figure~\ref{fig:Gaussian}. 

The authors of~\cite{dasarathy2015sketching} considered the recovery of an unknown sparse matrix $\mathbf{M}\in\mathbbm{R}^{\nn\times\nn}$ (not necessarily symmetric) from an $\ns\times\ns$ matrix $\overline{\mathbf{A}}=\overline{\mathbf{B}}\mathbf{M}\overline{\mathbf{C}}^{T}$ where $\overline{\mathbf{B}}\in\mathbbm{R}^{\ns\times\nn}$ and $\overline{\mathbf{C}}\in\mathbbm{R}^{\ns\times\nn}$ with $\ns\ll\nn$. By adding a symmetry constraint to their recovery formulation, we obtain the following modified basis pursuit as a convex optimization:
\begin{align}
\label{eq:basis_pursuit_1}
\mathrm{minimize} \ &||\mathbf{Y}(G)||_{1}\\
\mathrm{subject}\ \mathrm{to}\  &\mathbf{B}\mathbf{Y}(G)= \mathbf{A},\\
\label{eq:basis_pursuit_2}
&\mathbf{Y}(G)\in\mathbbm{S}^{\nn\times\nn}
\end{align}
where $||\mathbf{Y}(G)||_{1} = ||\mathrm{vec}(\mathbf{Y}(G))||_{1}$ is the entry-wise $\ell_1$-norm of $\mathbf{Y}(G)$ and $\mathbbm{S}^{\nn\times\nn}$ denotes the set of all symmetric matrices in $\mathbbm{R}^{\nn\times\nn}$. However, the approach in~\cite{dasarathy2015sketching} does not carry through to our setting for two reasons. First, the analysis of such an optimization often requires stronger assumptions, \textit{e.g.,} the non-zeros are not
concentrated in any single column (or row) of $\mathbf{Y}(G)$, as in~\cite{dasarathy2015sketching}.  Second, having the symmetry property of $\mathbf{Y}$ as a constraint does not explicitly make use of the fact that many columns in $\mathbf{Y}$ are indeed sparse and can be recovered correctly. As a consequence, basis pursuit may produce poor results in certain scenarios where our approach performs well, as demonstration in our experimental results on star graphs in Section~\ref{sec:basis_pursuit}.

Although the columns of $\mathbf{Y}(G)$ are correlated because of the symmetry, in general there are no constraints on the support sets of the columns. Thus distributed compressed sensing schemes (for instance,~\cite{baron2005distributed} requires the columns to share the same support set) are not directly applicable in this situation. 

The previous studies and aforementioned issues together motivate us to propose a novel three-stage recovery scheme for the derivation of a sufficient recovery condition, which leads to a practical algorithm that is sample and computationally efficient as well as robust to noise.
\subsection{Our Contributions}
\label{sec:con}
We demonstrate that the linear system in (\ref{eq:linear}) can be used to learn the topology and parameters of a graph. Our framework can be applied to perform system identification in electrical grids by leveraging synchronous nodal current and voltage measurements obtained from phasor measurement units (PMUs). 

Compared to existing methods and analysis, the main results of this paper are three-fold:
\begin{enumerate}
	\item {\textit{Fundamental Trade-offs}:}
	In Theorem~\ref{thm:1}, we derive a general lower bound on the \textit{probability of error} for topology identification (defined in~(\ref{poet})).
    In Section~\ref{sec:ach}, we describe a simple three-stage recovery scheme combining $\ell_1$-norm minimization with an additional step called \textit{consistency-checking}, rendering which allows us to bound the number of measurements for exact recovery from above as in Theorem~\ref{thm:3}.
    \item {\textit{(Worst-case) Sample Complexity}:} 
    We provide sample complexity results for recovering a random graph that may contain \textit{high-degree} nodes. The unknown distribution that the graph is sampled from is characterized based on the definition of ``\textit{$({\mu},{K},\rho)$-sparsity}" (see Definition~\ref{def:sparse}).
    Under the assumption that the matrix $\mathbf{B}$ has Gaussian IID entries, in Section~\ref{sec:GIPM}, we provide upper and lower bounds on the worst-case sample complexity in Theorem~\ref{thm:4}. We show two applications of Theorem~\ref{thm:4} for the uniform sampling of trees and the Erd\H{o}s-R\'{e}nyi $(\nn,p)$ model in Corollary~\ref{coro:1} and~\ref{coro:2}, respectively. 
	\item {\textit{(Heuristic) Algorithm}:} Motivated by the three-stage recovery scheme, a heuristic algorithm with  polynomial (in $\nn$) running-time is reported in Section~\ref{sec:6}, together with simulation results for power system test cases validating its performance in Section~\ref{sec:sim}.
\end{enumerate}

Some comments about the above results are as follows:

\subsection{Outline of the Paper}
The remaining content is organized as follows. In Section~\ref{sec:pre}, we specify our models. In Section~\ref{sec:fun}, we present the converse result as fundamental limits for recovery. The achievability is provided in ~\ref{sec:ach}. We present our main result as the worst-case sample complexity for Gaussian IID measurements in Section~\ref{sec:GIPM}. A heuristic algorithm together with simulation results are reported in Section~\ref{sec:6} and~\ref{sec:sim}.

\section{Model and Definitions}
\label{sec:pre}

\subsection{Notation}
Let $\mathbbm{F}$ denote a field that can either be the set of real numbers $\mathbbm{R}$, or the set of complex numbers $\mathbbm{C}$. The set of all symmetric $\nn\times\nn$ matrices whose entries are in $\mathbbm{F}$ is denoted by $\mathbbm{S}^{\nn\times\nn}$. The imaginary unit is denoted by $\mathrm{j}$. Throughout the work, let $\log\left(\cdot\right)$ denote the binary logarithm with base $2$ and let $\ln\left(\cdot\right)$ denote the natural logarithm with base $e$. We use $\mathbbm{E}\left[\cdot \right]$ to denote the expectation of random variables. The mutual information is denoted by $\mathbbm{I}(\cdot)$. The entropy function (either differential or discrete)  is denoted by $\mathbbm{H}(\cdot)$ and in particular, we reserve $h(\cdot)$ for the binary entropy function. To distinguish random variables and their realizations, we follow the convention and denote the former by capital letters (\textit{e.g.,} $A$) and the latter by lower case letters (\textit{e.g.,} $a$).  The symbol $C$ is used to designate a constant. 

Matrices are denoted in boldface (\textit{e.g.,} $\mathbf{A}$, $\mathbf{B}$ and $\mathbf{Y}$). The $i$-th row, the $j$-th column and the $(i,j)$-th entry of a matrix $\mathbf{A}$ are denoted by $A^{(i)}$, $A_j$ and $A_{i,j}$ respectively. For notational convenience, let $\mathcal{S}$ be a subset of $\mathcal{V}$. Denote by $\overline{\mathcal{S}}:=\mathcal{V}\backslash\mathcal{S}$ the complement of $\mathcal{S}$ and by 
$\mathbf{A}_{\mathcal{S}}$ 
a sub-matrix consisting of $|\mathcal{S}|$ columns of the matrix $\mathbf{A}$ whose indices are chosen from $\mathcal{S}$. The notation $\top$ denotes the transpose of a matrix,  $\mathrm{det}\left(\cdot\right)$ calculates its determinant.  For the sake of notational simplicity, we use big O notation ($o$,$\omega$,$O$,$\Omega$,$\Theta$) to quantify asymptotic behavior.


\subsection{Graphical Model}
\label{sec:model}
Denote by $\mathcal{V}=\{1,\ldots,\nn\}$ a set of $\nn$ nodes
and consider an \textit{undirected} graph $G=(\mathcal{V},\mathcal{E})$ (with no self-loops) whose edge set $\mathcal{E}\subseteq \mathcal{V}\times\mathcal{V}$ contains the desired topology information. The degree of each node $j$ is denoted by $d_j$. The connectivity between the nodes is unknown and our goal is to determine it by learning the associated \textit{graph matrix} using linear measurements. 

\begin{definition}[Graph matrix]
\label{def:graph_matrix}
Provided with an underlying graph $G=(\mathcal{V},\mathcal{E})$, a {\emph{symmetric}} matrix $\mathbf{Y}(G)\in\mathbbm{S}^{\nn\times\nn}$ is called a {\emph{graph matrix}} if the following conditions hold:
\begin{align*}
{Y}_{i,j}(G)=\begin{cases}
\neq 0 \quad &\text{ if } i\neq j \text{ and } (i,j)\in\mathcal{E}\\
0    \quad &\text{ if } i\neq j \text{ and } (i,j)\notin\mathcal{E}\\
\text{arbitrary} \quad &\text{ otherwise }
\end{cases}.
\end{align*}
\end{definition}
\begin{remark}
Our theorems can be generalized to recover a broader class of symmetric matrices, as long as the matrix to be recovered satisfies (1) Knowing $\mathbf{Y}(G)\in\mathbbm{F}^{\nn\times\nn}$ gives the full knowledge of the topology of $G$; (2) The number of non-zero entries in a column of $\mathbf{Y}(G)$ has the same order as the degree of the corresponding node, \textit{i.e.,} 
$|\mathrm{supp}({Y}_{j})| = O(d_j)$. for all $j\in\mathcal{V}$. To have a clear presentation, we consider specifically the case $|\mathrm{supp}({Y}_{j})| = d_j$.
\end{remark}
In this work, we employ a probabilistic model and assume that the graph $G$ is chosen randomly from a \textit{candidacy set} $\mathsf{C}(\nn)$ (with $\nn$ nodes), according to some distribution ${\mathcal{G}_\nn}$. Both the candidacy set $\mathsf{C}(\nn)$ and distribution ${\mathcal{G}_\nn}$ are not known to the estimator. For simplicity, we often omit the subscripts of $\mathsf{C}(\nn)$ and $\mathcal{G}_\nn$.

\begin{example}
\label{example}
We exemplify some possible choices of the candidacy set and distribution:
\begin{enumerate}

\item[\text{(a)}] \textit{{(Mesh Network)}}
When $G$ represents a transmission (mesh) power network and no prior information is available, the corresponding candidacy set $\mathsf{G}(\nn)$ consisting of all graphs with $\nn$ nodes and $G$ is selected uniformly at random from $\mathsf{G}(\nn)$. Moreover, $\left|\mathsf{G}(\nn)\right|=2^{\nn\choose 2}$ in this case.

\item[\text{(b)}] \textit{{(Radial Network)}}
When $G$ represents a distribution (radial) power network and no other prior information is available, then the corresponding candidacy set $\mathsf{T}(\nn)$ is a set containing all spanning trees of the complete graph with $\nn$ buses (nodes) and $G$ is selected uniformly at random from $\mathsf{T}(\nn)$; the cardinality is $\left|\mathsf{T}(\nn)\right|=\nn^{\nn-2}$ by Cayley's formula.

\item[\text{(c)}] \textit{{(Radial Network with Prior Information)}}
When $G=\left(\mathcal{V},\mathcal{E}\right)$ represents a distribution (radial) power network, and we further know that some of the buses cannot be connected (which may be inferred from locational/geographical information), then the corresponding candidacy set $\mathsf{T}_{H}(\nn)$ is a set of spanning trees of a sub-graph $H=\left(\mathcal{V},\mathcal{E}_H\right)$ with $\nn$ buses. An edge $e\notin\mathcal{E}_H$ if and only if we know $e\notin\mathcal{E}$. The size of $\mathsf{T}_{H}(\nn)$ is given by Kirchhoff's matrix tree theorem (\textit{c.f.}~\cite{west2001introduction}).

\item[\text{(d)}] \textit{{(Erd\H{o}s-R\'{e}nyi $(\nn,p)$ model)}}
In a more general setting, $G$ can be a random graph chosen from an ensemble of graphs according to a certain distribution. When a graph $G$ is sampled according to the Erd\H{o}s-R\'{e}nyi $(\nn,p)$ model, each edge of $G$ is connected IID with probability $p$. We denote the corresponding graph distribution for this case by $\mathcal{G}_{\mathrm{ER}}(\nn,p)$.
\end{enumerate}
\end{example}

The next section is devoted to describing available measurements. 
\subsection{Linear System of Measurements}
\label{sec:cvm}

Suppose the measurements are sampled discretely and indexed by the elements of the set $\{1,\ldots,\ns\}$. As a general framework, the measurements are collected in two matrices $\mathbf{A}$ and $\mathbf{B}$ and defined as follows.
\begin{definition}[{Generator and measurement matrices}]
Let $\ns$ be an integer with $1\leq \ns\leq \nn$. The \textit{generator matrix} $\mathbf{B}$ is an $\ns\times \nn$ \textit{random} matrix and  
the \emph{measurement matrix} $\mathbf{A}$ is an $\ns\times \nn$ matrix with entries selected from $\mathbbm{F}$ that satisfy the linear system (\ref{eq:linear}):
\begin{align*}
\mathbf{A} = \mathbf{B}\mathbf{Y}(G)+\mathbf{Z}
\end{align*}
where $\mathbf{Y}(G)\in\mathbbm{S}^{\nn\times \nn}$ is a graph matrix to be recovered, with an underlying graph $G$ and $\mathbf{Z}\in\mathbbm{F}^{m\times n}$ denotes the random \textit{additive  noise}. We call the recovery \textit{noiseless} if $\mathbf{Z}=\mathbf{0}$. Our goal is to resolve the matrix $\mathbf{Y}(G)$ based on given matrices $\mathbf{A}$ and $\mathbf{B}$. 
\end{definition}
 In the remaining contexts, we sometime simplify the matrix $\mathbf{Y}(G)$ as $\mathbf{Y}$ if there is no confusion.


\subsection{Applications to Electrical Grids}

Various applications fall into the framework in (\ref{eq:linear}). Here we present two examples of the graph identification problem in power systems. The measurements are modeled as time series data obtained via nodal sensors at each node, \textit{e.g.,} PMUs, smart switches, or smart meters. 
\subsubsection{Example $1$: Nodal Current and Voltage Measurements}
We assume data is obtained from a short time interval over which the unknown parameters in the network are \textit{time-invariant}.
$\mathbf{Y}\in\mathbbm{C}^{\nn\times \nn}$ denotes the \textit{nodal admittance matrix} of the network and is defined
\begin{align}
\label{eq:admi}
Y_{i,j}:=\begin{cases}
-y_{i,j} \quad &\text{if } i\neq j\\
y_i + \sum_{k\neq i} y_{i,k} & \text{if } i=j
\end{cases}
\end{align}
where $y_{i,j}\in\mathbbm{C}$ is the admittance of line $(i,j)\in\mathcal{E}$ and $y_i$ is the self-admittance of bus $i$. Note that if two buses are not connected then $Y_{i,j}=0$. 

The corresponding generator and measurement matrices are formed by simultaneously measuring both current (or equivalently, power injection) and voltage at each node and at each time step. For each $t=1,\ldots,\ns$, the nodal current injection is collected in an $\nn$-dimensional random vector $I_{t}=(I_{t,1},\ldots,I_{t,\nn})$.
Concatenating the $I_{t}$ into a matrix we get $\mathbf{I}:=[I_{1},I_{2},\ldots,I_{\ns}]^{\top}\in \mathbbm{C}^{\ns \times \nn}$.
The generator matrix  $\mathbf{V}:=[V_{1},V_{2},\ldots,V_{\ns}]^{\top}\in \mathbbm{C}^{\ns \times \nn}$ is constructed analogously. 
Each pair of measurement vectors $(I_{t},V_{t})$ from $\mathbf{I}$ and $\mathbf{V}$ must satisfy Kirchhoff's and Ohm's laws,
\begin{align}
\label{2.3}
{I}_{t}=\mathbf{Y}{V}_{t}, \quad t=1,\ldots,\ns.
\end{align}
In matrix notation (\ref{2.3}) is equivalent to $
\mathbf{I}=\mathbf{V}\mathbf{Y}
$, which is a noiseless version of the linear system defined in (\ref{eq:linear}).

Compared with only obtaining one of the current, power injection or voltage measurements (for example, as in~\cite{tan2011sample,tan2010learning,liao2015distribution}), collecting simultaneous current-voltage pairs doubles the amount of data to be acquired and stored. There are benefits however. First, exploiting the physical law relating voltage and current not only enables us to identify the topology of a power network but also recover the parameters of the admittance matrix. Furthermore, 
dual-type measurements significantly reduce the sample complexity for learning the graph, compared with the results for single-type measurements. 

\subsubsection{Example $2$: Nodal Power Injection and Phase Angles}
Similar to the previous example, at each time $t=1,\ldots,\ns$, denote by $P_{t,j}$ and $\theta_{t,j}$ the active nodal power injection and the phase of voltage at node $j$ respectively. The matrices $\mathbf{P}\in\mathbbm{R}^{\ns\times\nn}$ and $\pmb{\theta}\in\mathbbm{R}^{\ns\times\nn}$ are constructed in a similar way by concatenating the vectors ${P}_t=(P_{t,1},\ldots,P_{t,\nn})$ and ${\theta}_{t}=(\theta_{t,1},\ldots,\theta_{t,\nn})$. The matrix representation of the DC power flow model can be expressed as a linear system $\mathbf{P} =\pmb{\theta}\mathbf{C}\mathbf{S}\mathbf{C}^{\top}$, which belongs to the general class represented in (\ref{eq:linear}). Here, the diagonal matrix $\mathbf{S}\in\mathbbm{R}^{|\mathcal{E}|\times|\mathcal{E}|}$ is the susceptence matrix whose $e$-th diagonal entry represents the susceptence on the $e$-th edge in $\mathcal{E}$ and $\mathbf{C}\in \{-1,0,1\}^{\nn\times |\mathcal{E}|}$ is the node-to-link incidence matrix of the graph. 
The vertex-edge incidence matrix\footnote{Although the underlying network is a directed graph, when considering the fundamental limit for topology identification, we still refer to the recovery of an undirected graph $G$.}  $\mathbf{C}\in \{-1,0,1\}^{\nn\times |\mathcal{E}|}$ is defined as
\begin{align*}
C_{j,e}:=
\begin{cases}
1, \quad &\text{if bus } j \text{ is the source of } e\\
-1, \quad &\text{if bus } j \text{ is the target of } e\\
0, \quad &\text{otherwise}
\end{cases}.
\end{align*}
Note that $\mathbf{C}\mathbf{S}\mathbf{C}^{\top}$ specifies both the network topology and the susceptences of power lines.

\subsection{Probability of Error as the Recovery Metric}
\label{sec:poe}
We define the error criteria considered in this paper. We refer to finding the edge set $\mathcal{E}$ of $G$ via matrices $\mathbf{A}$ and $\mathbf{B}$ as the \textit{topology identification problem} and recovering the graph matrix $\mathbf{Y}$ via matrices $\mathbf{A}$ and $\mathbf{B}$ as the \textit{parameter reconstruction problem}.

\begin{definition}
Let $f$ be a function or algorithm that returns an estimated graph matrix $\mathbf{X}=f(\mathbf{A},\mathbf{B})$ given inputs $\mathbf{A}$ and $\mathbf{B}$. 
The \textit{probability of error for topology identification} $\poet$ is defined to be the probability that the estimated edge set is not equal to the correct edge set:
\begin{align}
\label{poet}
\poet:=
\mathbbm{P}\left(\exists \ i\neq j \ \big| \ \mathrm{sign}(X_{i,j})\neq \mathrm{sign}\left({Y}_{i,j}(G)\right)\right)
\end{align}
where the probability is taken over the randomness in $G,\mathbf{B}$ and $\mathbf{Z}$.
The \textit{probability of error for parameter reconstruction} $\poe(\eta)$ is defined to be the probability that the Frobenius norm of the difference between the estimate $\mathbf{X}$ and the original graph matrix $\mathbf{Y}(G)$ is larger than $\eta>0$: 
\begin{align}
\label{poe}
\poe(\eta):=\sup_{\mathbf{Y}\in\mathsf{Y}(G)}\mathbbm{P}\left(||\mathbf{X} -\mathbf{Y}(G)||_{\mathrm{F}}>\eta\right)
\end{align}
where $||\cdot||_{\mathrm{F}}$ denotes the Frobenius norm, $\eta>0$ and $\mathsf{Y}(G)$ is the set of all graph matrices $Y(G)$ that satisfy Definition \ref{def:graph_matrix} for the underlying graph $G$, and the probability is taken over the randomness in $G$, $\mathbf{B}$ and $\mathbf{Z}$. Note that for noiseless parameter reconstruction, \textit{i.e.,} $\mathbf{Z}=0$, we always consider exact recovery and set $\eta=0$ and abbreviate the probability of error as $\poe$.
\end{definition}



\section{Fundamental Trade-offs}

We discuss fundamental trade-offs of the parameter reconstruction problem defined in Section~\ref{sec:model} and~\ref{sec:cvm}. The converse result is summarized in Theorem~\ref{thm:1} as an inequality involving the probability of error, the distributions of the underlying graph, generator matrix and noise. Next, in Section~\ref{sec:ach}, we focus on a particular three-stage scheme, and show in Theorem~\ref{thm:3} that under certain conditions, the probability of error is asymptotically zero (in $\nn$). 
\subsection{Necessary Conditions}
\label{sec:fun}
The following theorem states the fundamental limit.
\begin{theorem}[\text{Converse}]
	\label{thm:1}
 The probability of error for topology identification $\poet$ is bounded from below as
	\begin{align}
	\label{3.0}
	\poet\geq 1-\frac{\mathbbm{H}\left(\mathbf{A}\right)-\mathbbm{H}\left(\mathbf{Z}\right)+\ln 2}{\mathbbm{H}\left({\mathcal{G}_\nn}\right)}
	\end{align}
where $\mathbbm{H}\left(\mathbf{A}\right)$, $\mathbbm{H}\left(\mathbf{Z}\right)$ are differential entropy (in base $e$) functions of the random variables $\mathbf{A}$, $\mathbf{Z}$ respectively and $\mathbbm{H}\left({\mathcal{G}_\nn}\right)$ is the entropy (in base $e$) of the probability distribution ${\mathcal{G}_\nn}$.
\end{theorem}

\begin{remark}
It can be inferred from the theorem that $\poet=1-O(\ns\nn/\mathbbm{H}\left({\mathcal{G}_\nn}\right))$, given that the generator matrix $\mathbf{B}$ has Gaussian IID entries and the noise $\mathbf{Z}$ is additive white Gaussian (see Lemma~\ref{lemma:2}). Therefore, the structure of the graphs reflected in the corresponding entropy of the graph distribution determines the number of samples needed. Consider the four cases listed in Example~\ref{example}.
The number of samples must be at least linear in $\nn$ (size of the graph) to ensure a small probability of error, given that the graph, as a mesh network, is chosen uniformly at random from $\mathsf{C}(\nn)$ (see Example~\ref{example} (a)) since $\mathbbm{H}(\mathcal{U}_{\mathsf{G}(\nn)})={\nn\choose 2}$. On the other hand, as corollaries, under the assumptions of Gaussian IID measurements, $\ns=\Omega(\log \nn)$ is \textit{necessary} for making the probability of error less or equal to $1/2$, if the graph is chosen uniformly at random from $\mathsf{T}(\nn)$; $\ns=\Omega(\nn h(p))$ is \textit{necessary} if the graph is sampled according to $\mathcal{G}_{\mathrm{ER}}(\nn,p)$, as in Examples~\ref{example} (b) and (c), respectively.
The theorem can be generalized to complex measurements by adding additional multiplicative constants.
\end{remark}

Note that $\poe\geq \poet$ for any fixed noiseless parameter reconstruction algorithm, the necessary conditions work for both topology and (noiseless) parameter reconstruction. The proof is postponed to Appendix~\ref{app:proof_fundamental_limits} and the key steps are first applying the generalized Fano's inequality~(see \cite{fano1961transmission,aeron2010information}) and then bounding the mutual information $\mathbbm{I}\left(G;\mathbf{A}|\mathbf{B}\right)$ from above by $\mathbbm{H}(\mathbf{A})-\mathbbm{H}(\mathbf{Z})$. The general converse stated in Theorem~\ref{thm:1} is used in asserting the results on worst-case sample complexity in Theorem~\ref{thm:4}. Next, we analyze the sufficient condition for recovering a graph matrix $\mathbf{Y}(G)$. Before proceeding to the results, we introduce a novel characterization of the distribution $\mathcal{G}_{\nn}$, from which a graph $G$ is sampled. In particular, the graph $G$ is allowed to have high-degree nodes.

\subsection{Characterization of Graph Distributions}
\label{sec:sparse_distribution}
Let $d_j(G)$ denote the degree of node $j\in\mathcal{V}$ in $G$. Denote by 
$
\mathcal{V}_{\mathrm{Large}}\left({\mu}\right):=\{j\in\mathcal{V} \ \big| \ d_j(G)>{\mu}\}
$ the set of nodes having degrees greater than the \textit{threshold parameter} $0\leq {\mu}\leq \nn-2$ and
$\mathcal{V}_{\mathrm{Small}}\left({\mu}\right):= \mathcal{V}\backslash\mathcal{V}_{\mathrm{Large}}\left({\mu}\right)$ the set of nodes for all $\mu$-sparse column vectors of $\mathbf{Y}$. With a \textit{counting parameter} $0\leq {K}\leq\nn$, we define a set of graphs wherein each graph consists of no more than $\nk$ nodes with degree larger than $\mu$, denoted by
$
\mathsf{C}(\nn,\mu,\nk):=\{G\in\mathsf{C}(\nn) \ | \left|\mathcal{V}_{\mathrm{Large}}\left({\mu}\right)\right|\leq \nk\}.
$
The following definition characterizes graph distributions.
\begin{definition}[$(\mu,K,\rho)$-sparse distribution]
\label{def:sparse}
A graph distribution $\mathcal{G}_{\nn}$ is said to be  $(\mu,K,\rho)$-\emph{sparse} if assuming that $G$ is distributed according to $\mathcal{G}_{\nn}$, then the probability that $G$ belongs to $\mathsf{C}(\nn,\mu,\nk)$ is larger than $1-\rho$, \textit{i.e.,}
\begin{align}
\label{eq:rec}
\mathbbm{P}_{{\mathcal{G}_{\nn}}}\left(G\notin\mathsf{C}(\nn,\mu,\nk)\right)\leq \rho.
\end{align}
\end{definition}

\textit{1) Uniform Sampling of Trees:}

Based on the definition above, for particular graph distributions, we can find the associated parameters. We exemplify by considering two graph distributions introduced in Example~\ref{example}.
Denote by $\mathcal{U}_{\mathsf{T}(\nn)}$ the uniform distribution on the set $\mathsf{T}(\nn)$ of all trees with $\nn$ nodes. 
\begin{lemma}
\label{lemma:trees}
For any ${\mu}\geq 1+\ln \nn$ and ${K}>0$, the distribution $\mathcal{U}_{\mathsf{T}(\nn)}$ is $({\mu },{K},1/K)$-{sparse}. 
\end{lemma}

\textit{2) Erd\H{o}s-R\'{e}nyi $(\nn,p)$ model:}

Denote by $\mathcal{G}_{\mathrm{ER}}(\nn,p)$ the graph distribution for the Erd\H{o}s-R\'{e}nyi $(\nn,p)$ model. Similarly, the lemma below classifies $\mathcal{G}_{\mathrm{ER}}(\nn,p)$ into a $({\mu },{\nk},\rho)$-{sparse} distribution with appropriate parameters.

\begin{lemma}
\label{lemma:renyi}
For any ${\mu}(\nn,p)$ that satisfies ${\mu}(\nn,p)\geq {2\nn h(p)}/{(\ln 1/p)}$ and ${K}>0$, the distribution $\mathcal{G}_{\mathrm{ER}}(\nn,p)$ is  $({\mu},{K},\nn\exp(-\nn h(p))/K)$-{sparse}. 
\end{lemma}
The proofs of Lemmas~\ref{lemma:trees} and~\ref{lemma:renyi}  are in Appendix~\ref{app:proof_trees}. 
\begin{remark}
It is worth noting that the $(\mu,K,\rho)$-\emph{sparsity} is capable of characterizing \textit{any} arbitrarily chosen distribution. The interesting part is that for some of the well-known distributions, such as $\mathcal{G}_{\mathrm{ER}}(\nn,p)$, this sparsity characterization offers a method that can be used in the analysis and moreover, it leads to an \textit{exact characterization} of sample complexity for the noiseless case. Therefore, for the particular examples presented in Lemma~\ref{lemma:trees} and Lemma~\ref{lemma:renyi}, the selected threshold and counting parameters for both of them are ``tight" (up to multiplicative factors), in the sense that the corresponding sample complexity matches (up to multiplicative factors) the lower bounds derived from Theorem~\ref{thm:1}. This can be seen in Corollary~\ref{coro:1} and~\ref{coro:2}.
\end{remark}

\label{sec:def_of_rip}
\begin{algorithm}[h]
	\label{alg:1}
	\hrule
	\medskip
	\KwData{Matrices of measurements $\mathbf{A}$ and $\mathbf{B}$}
	\KwResult{Estimated graph matrix $\mathbf{X}$}

	\medskip
	\hrule
	\medskip
	\textit{\text{Step (a): Recovering columns independently}}: 
	\medskip
	\hrule	
	\medskip
	\For{$j\in\mathcal{V}$}{Solve the following $\ell_1$-minimization and obtain an optimal $\mathbf{X}$:
		\begin{align}
		\nonumber
		\mathrm{minimize}  \quad &{\big|\big|X_j\big|\big|_{1}}\\
		\nonumber
		\mathrm{subject}\ \mathrm{to}\  & ||\mathbf{B}X_{j} - A_j||_2\leq \gamma,\\
		\nonumber
		&X_j\in\mathbbm{F}^{\nn}.
		\end{align}}
	\medskip
	\hrule
	\medskip
	\textit{\text{Step (b): Consistency-checking}}:
	\medskip
	\hrule
	\medskip	
	\For{$\mathcal{S}\subseteq\mathcal{V}$ with $|\mathcal{S}|=\nn-{K}$}{
			\If{$|X_{i,j} - X_{j,i}|> 2\gamma$ for some $i,j\in\mathcal{S}$}
			{$\textbf{continue}$\;}
		\Else{
		\For{$j\in\overline{\mathcal{S}}$}{
	\medskip
	\hrule
	\medskip
	\textit{\text{Step (c): Resolving unknown entries}}:
	\medskip
	\hrule
	\medskip
			Update $X_{j}^{\overline{\mathcal{S}}}$ by solving the linear system:
			\begin{align}
			\nonumber
			\mathbf{B}_{\overline{\mathcal{S}}}X_{j}^{\overline{\mathcal{S}}}=A_j- \mathbf{B}_{\mathcal{S}}X^{\mathcal{S}}_{j}.
			\end{align}}}
		\textbf{break;}
	}	
			\Return{$\mathbf{X}=(X_1,\ldots,X_\nn)$}\;
	\medskip
	\hrule
	\hrule
	\medskip
		\caption{A three-stage recovery scheme. The first stage focuses on solving each column of the matrix $\mathbf{Y}$ independently using $\ell_1$-minimization. In the second stage, the recovery correctness of the first stage is further verified via \textit{consistency-checking}, which utilizes the fact that the matrix to be recovered $\mathbf{Y}$ is \textit{symmetric}. The parameter $\gamma$ is set to zero for the analysis of noiseless parameter reconstruction.}
\end{algorithm}

\subsection{Sufficient Conditions}

\label{sec:ach}
In this subsection, we consider the sufficient conditions (achievability) for parameter reconstruction. The proofs rely on constructing a three-stage recovery scheme (Algorithm~\ref{alg:1}), which contains three steps -- \textit{column-retrieving}, \textit{consistency-checking} and \textit{solving unknown entries}. The worst-case running time of this scheme depends on the underlying distribution $\mathcal{G}_\nn$\footnote{Although for certain distributions, the computational complexity is not polynomial in $\nn$, the scheme still provides insights on the fundamental trade-offs between the number of samples and the probability of error for recovering graph matrices. Furthermore, motivated by the scheme, a polynomial-time heuristic algorithm is provided in Section~\ref{sec:6} and experimental results are reported in Section~\ref{sec:sim}.}. 
 The scheme is presented as follows.
 
 \begin{figure}[h!]
	\centering
	\includegraphics[scale=0.35]{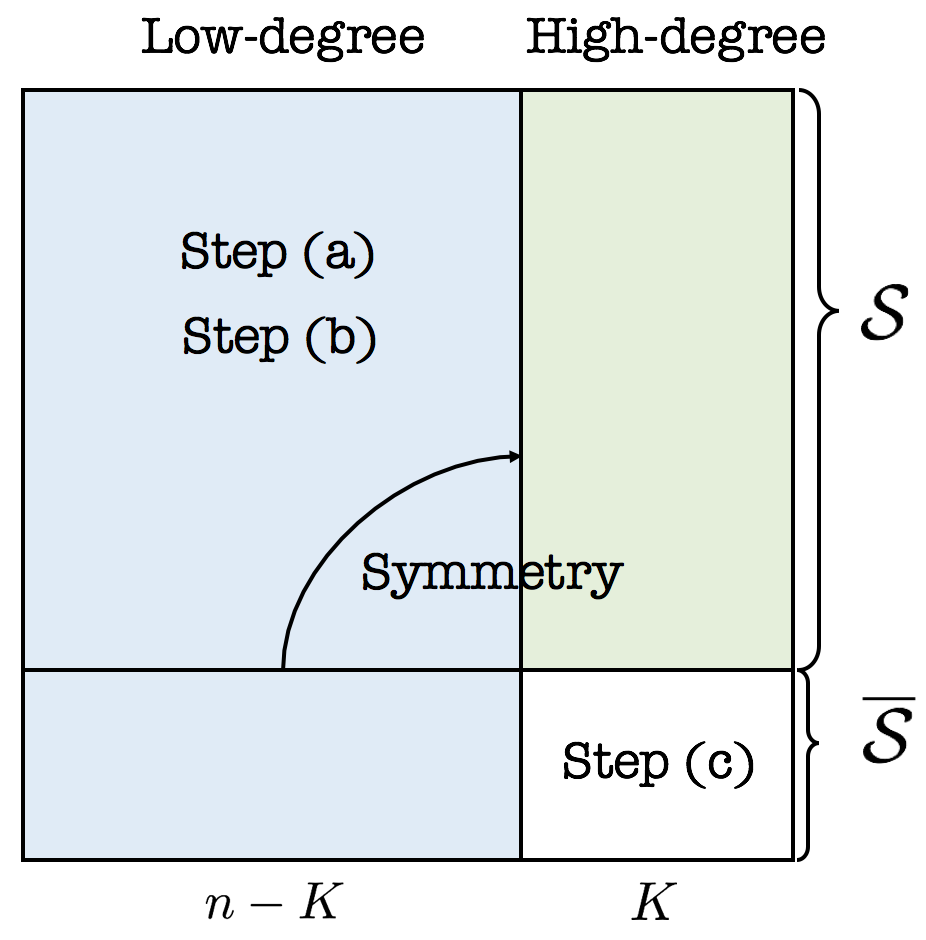}
	\caption{The recovery of a graph matrix $\mathbf{Y}$ using the three-stage scheme in Algorithm~\ref{alg:1}. The $\nn-\nk$ columns of $\mathbf{Y}$ colored by gray are first recovered via the $\ell_1$-minimization~\eqref{4.1}-\eqref{4.2} in step (a), after they are accepted by passing the consistency check in step (b). Then, symmetry is used for recovering the entries in the matrix marked by green. Leveraging the linear measurements again, in step (c), the remaining $\nk^2$ entries in the white symmetric sub-matrix are solved using Equation~\eqref{4.3}.}
	\label{fig:concept_1}
	\medskip
\end{figure}
 
\textit{1) Three-stage Recovery Scheme:}

\textit{Step (a): {Retrieving columns.}}
In the first stage, using $\ell_1$-norm minimization, we recover each column of $\mathbf{Y}$ based on ~(\ref{eq:linear}):
\begin{subequations}
\begin{align}
\label{4.1}
\mathrm{minimize}  \quad &{\big|\big|X_j\big|\big|_{1}}\\
\label{4.0}
\mathrm{subject}\ \mathrm{to}\quad  &||\mathbf{B}X_{j} - A_j||_2\leq \gamma,\\
\label{4.2}
&X_j\in\mathbbm{F}^{\nn}.
\end{align}	
\end{subequations}
Let
${X}_{j}^{\mathcal{S}}:=({X}_{i,j})_{i\in\mathcal{S}}$
be a length-$|\mathcal{S}|$ column vector consisting of $\left|\mathcal{S}\right|$ coordinates in ${X}_j$, the $j$-th retrieved column. We do not restrict the methods for solving the $\ell_1$-norm minimization in (\ref{4.1})-(\ref{4.2}), as long as there is a unique solution for sparse columns with fewer than ${\mu}$ non-zeros (provided enough number of measurements and the parameter ${\mu}>0$ is defined in Definition~\ref{def:sparse}).

\textit{Step (b): Checking consistency.}

In the second stage, we check for error in the decoded columns $X_1,\ldots,X_\nn$ using the symmetry property (perturbed by noise) of the graph matrix $\mathbf{Y}$. Specifically, we fix a subset $\mathcal{S}\subseteq\mathcal{V}$ with a given size 
$\left|\mathcal{S}\right|=\nn-{K}$ for some integer\footnote{The choice of $\nk$ depends on the structure of the graph to be recovered and more specifically, $\nk$ is the counting parameter in Definition~\ref{def:sparse}. In Theorem~\ref{thm:3} and Corollary~\ref{corollary:noiseless}, we analyze the sample complexity of this three-stage recovery scheme by characterizing an arbitrary graph into the classes specified by Definition~\ref{def:sparse} with a fixed $\nk$.} $0\leq {K}\leq\nn$. Then we check if $|X_{i,j}- X_{j,i}|\leq 2\gamma$ for all $i,j\in\mathcal{S}$. If not, we choose a different set $\mathcal{S}$ of the same size. This procedure stops until either we find such a subset $\mathcal{S}$ of columns, or we go through all possible subsets without finding one.
In the latter case, an error is declared and the recovery is unsuccessful. It remains to recover the vectors $X_j$ for $j\in\overline{\mathcal{S}}$.

\textit{Step (c): Resolving unknown entries.}
In the former case, for each vector $X_j, j\in\mathcal{S}$, we accept its entries $X_{i,j}, i\in\overline{\mathcal{S}}$, as correct and therefore, according to the symmetry assumption, we know the entries $X_{i,j}, i\in{\mathcal{S}}, j\in\overline{\mathcal{S}}$ (equivalently $\{X_j^{\mathcal{S}}:j\in\overline{\mathcal{S}}\}$), which are used together with the sub-matrices $\mathbf{B}_{\mathcal{S}}$ and $\mathbf{B}_{\overline{\mathcal{S}}}$ to compute the other entries $X_{i,j}, i\in\overline{\mathcal{S}}$, of $X_j$ using~(\ref{4.0}):
\begin{align}
\label{4.3}
\mathbf{B}_{\overline{\mathcal{S}}}X_{j}^{\overline{\mathcal{S}}}=A_j- \mathbf{B}_{\mathcal{S}}X^{\mathcal{S}}_{j}, \quad j\in\overline{\mathcal{S}}.
\end{align}
Note that to avoid being over-determined, in practice, we solve
\begin{align}
\nonumber
\mathbf{B}^{\mathcal{K}}_{\overline{\mathcal{S}}}X_{j}^{\overline{\mathcal{S}}}=A^{\mathcal{K}}_j- \mathbf{B}^{\mathcal{K}}_{\mathcal{S}}X^{\mathcal{S}}_{j}, \quad j\in\overline{\mathcal{S}}
\end{align}
where $\mathbf{B}^{\mathcal{K}}_{\overline{\mathcal{S}}}$ is a $\nk\times\nk$ matrix whose rows are selected from $\mathbf{B}_{\overline{\mathcal{S}}}$ corresponding to $\mathcal{K}\subseteq{\mathcal{V}}$ with $|\mathcal{K}|=\nk$ and $\mathbf{B}^{\mathcal{K}}_{\mathcal{S}}$ selects the rows of $\mathbf{B}_{\mathcal{S}}$ in the same way.
We combine $X_j^{\mathcal{S}}$ and $X_j^{\overline{\mathcal{S}}}$ to obtain a new estimate $X_j$ for each $j\in\overline{\mathcal{S}}$. Together with the  columns $X_j$, $j\in\mathcal{S}$, that we have accepted, they form the estimated graph matrix $\mathbf{X}$. 
We illustrate the three-stage scheme in Figure~\ref{fig:concept_1}. 
In the sequel, we analyze the sample complexity of the three-stage scheme  based on the $({\mu},{K},\rho)$-sparse distributions defined in Definition~\ref{def:sparse}.

\textit{2) Analysis of the Scheme:}

Let $\mathbbm{F}\equiv\mathbbm{R}$ for the simplicity of representation and analysis. We now present another of our main theorems. Consider the models defined in Section~\ref{sec:model} and~\ref{sec:cvm}. 
The $\Gamma$-\textit{probability of error} is defined to be the maximal probability that the $\ell_2$-norm of the difference between the estimated vector $X\in\mathbbm{R}^{\nn}$ and the original vector $Y\in\mathbbm{R}^{\nn}$ (satisfying $A = \mathbf{B}Y+Z$ and both $A$ and $\mathbf{B}$ are known to the estimator) is larger than $\Gamma>0$:
\begin{align}
\nonumber
\overline{\varepsilon}_{\mathrm{P}}(\Gamma):=
\sup_{Y\in\mathsf{Y}(\mu)}\mathbbm{P}\left(\left|\left|X - Y\right|\right|_2> \Gamma\right)
\end{align}
where $\mathsf{Y}(\mu)$ is the set of all $\mu$-sparse vectors in $\mathbbm{R}^{\nn}$ and the probability is taken over the randomness in the generator matrix $\mathbf{B}$ and the additive noise $Z$.
%
Given a generator matrix $\mathbf{B}$, the corresponding \textit{restricted isometry constant} denoted by $\delta_{\mu}$ is the smallest positive number with
\begin{align}
\label{eq:ric}
\left(1-\delta_{\mu}\right)\left|\left|\mathbf{x}\right|\right|_{2}^2 \leq \left|\left|\mathbf{B}_{\mathcal{S}}\mathbf{x}\right|\right|_{2}^2 \leq \left(1+\delta_{\mu}\right)\left|\left|\mathbf{x}\right|\right|_{2}^2
\end{align}
for all subsets $\mathcal{S}\subseteq\mathcal{V}$ of size $\left|\mathcal{S}\right|\leq {\mu}$ and all $\mathbf{x}\in\mathbbm{R}^{\left|\mathcal{S}\right|}$.
Below we state a sufficient condition\footnote{Note that $\gamma$ cannot be chosen arbitrarily and $\Gamma$ depends on $\gamma$; otherwise the probability of error $\overline{\varepsilon}_{\mathrm{P}}(\Gamma)$ will blow up. Theorem~\ref{thm:noisy_sample_complexity} indicates that for Gaussian ensembles setting $\Gamma=O(\gamma)=O(\sqrt{\nn}\sigma_{\mathrm{N}})$ is a valid choice where $\sigma_{\mathrm{N}}$ is the standard deviation of each independent $Z_{i,j}$ in $\mathbf{Z}$.} derived form the three-stage scheme for parameter reconstruction.

\begin{theorem}[Achievability]
\label{thm:3}
Suppose the generator matrix satisfies that $\mathbf{B}^{\mathcal{K}}_{\overline{\mathcal{S}}}\in\mathbbm{R}^{\nk\times\nk}$ is invertible for all $\overline{\mathcal{S}}\subseteq\mathcal{V}$ and $\mathcal{K}\subseteq\mathcal{V}$ with $|\overline{\mathcal{S}}|=|\mathcal{K}|=\nk$. Let the distribution $\mathcal{G}_{\nn}$ be $({\mu},{K},\rho)$-sparse.
If the three-stage scheme in Algorithm~\ref{alg:1} is used for recovering a graph matrix $\mathbf{Y}(G_n)$ of $G_n$ that is sampled according to $\mathcal{G}_\nn$, then the probability of error satisfies $\poe(\eta)\leq \rho + (\nn-\nk)\overline{\varepsilon}_{\mathrm{P}}(\Gamma)$ with $\eta$ greater or equal to
\begin{align*}
2\left(\nn\Gamma+\frac{\Gamma||\mathbf{B}||_2+\gamma}{1-\delta_{2\nk}} \right)\big(2(\nn-\nk)+\nk\xi(\mathbf{B})\big)
\end{align*}
where $\delta_{2\nk}$ is the corresponding restricted isometry constant of $\mathbf{B}$ with $\mu=2\nk$ defined in~\eqref{eq:ric} and
$$
\xi(\mathbf{B}):=\max_{{\mathcal{S}},\mathcal{K}\subseteq{\mathcal{V}},|\overline{\mathcal{S}}|=|\mathcal{K}|=\nk}\big|\big|\mathbf{B}_{\mathcal{S}}\big|\big|_2\big|\big|\left(\mathbf{B}^{\mathcal{K}}_{\overline{\mathcal{S}}}\right)^{-1}\big|\big|_2.
$$
\end{theorem}
The proof is in Appendix~\ref{app:proof_achieve}. 
The theory of classical compressed sensing (see~\cite{candes2005error,rudelson2008sparse,candes2006near}) implies that for noiseless parameter reconstruction, if the generator matrix $\mathbf{B}$ has restricted isometry constants $\delta_{2{\mu}}$ and $\delta_{3{\mu}}$ satisfying $\delta_{2{\mu}}+\delta_{3{\mu}}< 1$, then all columns $Y_j$ with  $j\in\mathcal{V}_{\mathrm{Small}}$ are correctly recovered using the minimization in (\ref{4.1})-(\ref{4.2}). Denote by $\mathrm{spark}(\mathbf{B})$ the smallest number of columns in the matrix $\mathbf{B}$ that are linearly dependent (see \cite{donoho2003optimally} for the requirements on the spark of the generator matrix to guarantee desired recovery criteria). The following corollary is an improvement of Theorem~\ref{thm:3} for the noiseless case. The proof is postponed to Appendix~\ref{app:proof_noiseless}.
\begin{corollary}
\label{corollary:noiseless}
Let $\mathbf{Z}=0$ and suppose the generator matrix $\mathbf{B}$ has restricted isometry constants $\delta_{2{\mu}}$ and $\delta_{3{\mu}}$ satisfying $\delta_{2{\mu}}+\delta_{3{\mu}}< 1$
and furthermore, $\mathrm{spark}(\mathbf{B})>2K$.
If the distribution $\mathcal{G}_{\nn}$ is $({\mu},{K},\rho)$-sparse, then the probability of error for the three-stage scheme to recover the parameters of a graph matrix $\mathbf{Y}(G_n)$ of $G_n$ that is sampled according to $\mathcal{G}_\nn$ satisfies $\poe\leq \rho$.
\end{corollary}

\section{Gaussian IID Measurements}
\label{sec:GIPM}
In this section, we consider a special regime when the measurements in the matrix $\mathbf{B}$ are Gaussian IID random variables. Utilizing the converse in Theorem~\ref{thm:1} and the achievability in Theorem~\ref{thm:3}, the Gaussian IID assumption allows the derivation of explicit expressions of sample complexity as upper and lower bounds on the number of measurements $\ns$.  Combining with the results in Lemma~\ref{lemma:trees} and~\ref{lemma:renyi}, we are able to show that for the corresponding lower and upper bounds match each other for graphs distributions $\mathcal{U}_{\mathsf{T}(\nn)}$ and $\mathcal{G}_{\mathrm{ER}}(\nn,p)$ (with certain conditions on $p$ and $\nn$).

For the convenience of presentation, in the remainder of the paper, we restrict that the measurements are chosen from $\mathbbm{R}$, although the theorems can be generalized to the complex measurements. 
In realistic scenarios, for instance, a power network, besides the measurements collected from the nodes, nominal state values, \textit{e.g.,} operating current and voltage measurements are known to the system designer a priori. Representing the nominal values at the nodes by $\overline{A}\in\mathbbm{R}^\nn$ and $\overline{B}\in\mathbbm{R}^\nn$ respectively, the measurements in $\mathbf{A}$ and $\mathbf{B}$ are centered around $\ns\times \nn$ matrices $\overline{\mathbf{A}}$ and $\overline{\mathbf{B}}$ defined as
\begin{align*}
\overline{\mathbf{A}}:=\begin{bmatrix}
\cdots & \overline{A} & \cdots\\
\cdots & \overline{A} & \cdots\\
& \vdots &\\
\cdots & \overline{A} & \cdots
\end{bmatrix},\quad
\overline{\mathbf{B}}:=\begin{bmatrix}
\cdots & \overline{B} & \cdots\\
\cdots & \overline{B} & \cdots\\
& \vdots &\\
\cdots & \overline{B} & \cdots
\end{bmatrix}.
\end{align*}
The rows in $\mathbf{A}$ and $\mathbf{B}$ are the same, because the graph parameters are time-invariant, so are the nominal values.
Without system fluctuations and noise, the nominal values satisfy the linear system in~(\ref{eq:linear}), \textit{i.e.,}
\begin{align}
\label{2.10}
\overline{\mathbf{A}} = \overline{\mathbf{B}}\mathbf{Y}.
\end{align}
Knowing $\overline{A}$ and $\overline{B}$ is not sufficient to infer the network parameters (the entries in the graph matrix $\mathbf{Y}$), since the rank of the matrix $\overline{B}$ is one. However, measurement fluctuations can be used to facilitate the recovery of $\mathbf{Y}$.
The deviations from the nominal values are denoted by additive perturbation matrices $\mathbf{\widetilde{A}}$ and $\mathbf{\widetilde{B}}$ such that
$
\mathbf{A}=\overline{\mathbf{A}}+\mathbf{\widetilde{A}}.
$
Similarly,
$
\mathbf{B}=\overline{\mathbf{B}}+\mathbf{\widetilde{B}}
$
where $\widetilde{\mathbf{B}}$ is an $\ns\times \nn$ matrix consisting of additive perturbations. Therefore, considering the original linear system in  (\ref{eq:linear}), the equations above imply that
$\overline{\mathbf{A}}+\mathbf{\widetilde{A}}=\mathbf{B}\mathbf{Y}+\mathbf{Z} =\overline{\mathbf{B}}\mathbf{Y}+\mathbf{\widetilde{B}}\mathbf{Y}+\mathbf{Z}$
leading to
$\mathbf{\widetilde{A}}= \mathbf{\widetilde{B}}\mathbf{Y}+\mathbf{Z}$
where we have made use of (\ref{2.10}) and extracted the perturbation matrices $\mathbf{\widetilde{A}}$ and $\mathbf{\widetilde{B}}$. We specifically consider the case when the additive perturbations $\mathbf{\widetilde{B}}$ is a matrix with Gaussian IID entries. Without loss of generality, we suppose the mean of the Gaussian random variable is zero and the standard deviation is $\sigma_{\mathrm{S}}$. We consider additive white Gaussian noise (AWGN) with mean zero and standard deviation  $\sigma_{\mathrm{N}}$.
For simplicity, in the remainder of this paper, we slightly abuse the notation and replace the perturbation matrices $\mathbf{\widetilde{A}}$ and $\mathbf{\widetilde{B}}$ by $\mathbf{A}$ and $\mathbf{B}$ (we assume that $\mathbf{{B}}$ is Gaussian IID), if the context is clear. Under the assumptions above, the following lemma can be inferred from Theorem~\ref{thm:1} and the proof is in Appendix~\ref{app:proof_of_Gaussian_converse}.
\begin{lemma}
\label{lemma:2}
Consider the linear model $\mathbf{A}=\mathbf{B}\mathbf{Y}+\mathbf{Z}$.
Suppose $B_{i,j}\sim\mathcal{N}(0,\sigma^2_{\mathrm{S}})$ and $Z_{i,j}\sim\mathcal{N}(0,\sigma^2_{\mathrm{N}})$ are mutually independent Gaussian random variables for all $i,j\in\mathcal{V}$.
The probability of error for topology identification $\poet$ is bounded from below as
\begin{align}
\label{eq:noisy_poe}
\poet\geq&1-\frac{\nn\ns\ln\left(1+\frac{\sigma^2_{\mathrm{S}}}{\sigma^2_{\mathrm{N}}}\overline{Y}\right)}{2\mathbbm{H}\left({\mathcal{G}_\nn}\right)}
\end{align}
where $\overline{Y}:=\max_{i,j}\left|Y_{i,j}\right|$ denotes the maximal absolute value of the entries in the graph matrix $\mathbf{Y}$. In particular, if $\mathbf{Z}=0$, then for parameter reconstruction,
\begin{align}
\label{4.20}
\poe\geq&1-\frac{\nn\ns\ln\left(2\pi e \overline{Y}\sigma^2_{\mathrm{S}}\right)}{2\mathbbm{H}\left({\mathcal{G}_\nn}\right)}.
\end{align}
\end{lemma}

\subsection{Sample Complexity for Sparse Distributions}
\label{sec:4.c}

We consider the worst-case sample complexity for recovering graphs generated according to a sequence of sparse distributions, defined similarly as Definition~\ref{def:sparse} to characterize asymptotic behavior of graph distributions.
\begin{definition}[Sequence of sparse distributions]
\label{def:sparse_seq}
A sequence $\{\mathcal{G}_{\nn}\}$ of graph distributions is said to be $(\mu,\nk)$-\emph{sparse} if assuming a sequence of graphs $\{G_\nn\}$ is generated according to  $\{\mathcal{G}_\nn\}$, the sequences $\{\mu(\nn)\}$ and $\{K(\nn)\}$ guarantee that
\begin{align}
\label{eq:rec_seq}
\lim_{\nn\rightarrow\infty}\mathbbm{P}_{{\mathcal{G}_\nn}}\left(G_\nn\notin\mathsf{C}(\nn)({\mu(\nn)},{K(\nn)})\right)=0.
\end{align}
\end{definition}
In the remaining contexts, we write $\mu(\nn)$ and $K(\nn)$ as $\mu$ and $K$ for simplicity if there is no confusion. Based on the sequence of sparse distributions we defined above, we show the following theorem, which provides upper and lower bounds on the worst-case sample complexity, with Gaussian IID measurements.
\begin{theorem}[Noiseless worst-case sample complexity]
\label{thm:4}
Let $\mathbf{Z}=0$.
Suppose that the generator matrix $\mathbf{B}$ has Gaussian IID entries with mean zero and variance one and assume ${\mu}<\nn^{-3/{\mu}}(\nn-{K})$ and ${K}=o(\nn)$. For any sequence of  distributions that is $({\mu},{K})$-sparse, the three-stage scheme guarantees that $\lim_{\nn\rightarrow\infty}\poe=0$ using $\ns=O\left({\mu}\log({\nn}/{{\mu}})+{K}\right)$ measurements. Conversely, there exists a $({\mu},{K})$-sparse sequence of distributions such that the number of measurements must satisfy 
$\ns=\Omega\left({\mu}\log({\nn}/{{\mu}})+{K}/\nn^{3/{\mu}}\right)$
to make the probability of error $\poe$ less than ${1}/{2}$ for all $\nn$.
\end{theorem}
The proof is postponed to Appendix~\ref{app:proof_sample}.
\begin{remark}
The upper bound on $\ns$ that we are able to show differs from the lower bound by a sub-linear term $n^{3/\mu}$. In particular, when the term $\mu\log(\nn/\mu)$ dominates $K$, the lower and upper bounds become tight up to a multiplicative factor.
\end{remark}

\subsection{Applications of Theorem~\ref{thm:4}}

\textit{1) Uniform Sampling of Trees:}

As one of the applications of Theorem~\ref{thm:4}, we characterize the sample complexity of the uniform sampling of trees.
\begin{corollary}
\label{coro:1}Let $\mathbf{Z}=0$.
Suppose that the generator matrix $\mathbf{B}$ has Gaussian IID entries with mean zero and variance one and assume $G_\nn$ is distributed according to $\mathcal{U}_{\mathsf{T}(\nn)}$. There exists an algorithm that guarantees $\lim_{\nn\rightarrow\infty}\poe=0$ using $\ns=O\left((\log\nn)^2\right)$ measurements. Conversely, the number of measurements must satisfy 
$\ns=\Omega\left(\log\nn\right)$
to make the probability of error $\poe$ less than ${1}/{2}$.
\end{corollary}

\begin{proof}
The achievability follows from combining Theorem~\ref{thm:4} and Lemma~\ref{lemma:trees}, by setting $K(\nn)=\log\nn$. Substituting $\mathbbm{H}(\mathcal{U}_{\mathsf{T}(\nn)})=\Omega\left(\nn\log\nn\right)$ into (\ref{4.20}) yields the desired result for converse.
\end{proof}

\textit{2) Erd\H{o}s-R\'{e}nyi $(\nn,p)$ model:}

Similarly, recalling Lemma~\ref{lemma:renyi}, the sample complexity for recovering a random graph generated according to the Erd\H{o}s-R\'{e}nyi $(\nn,p)$ model is obtained.
\begin{corollary}
\label{coro:2}
Let $\mathbf{Z}=0$.
Assume $G_\nn$ is a random graph sampled according to $\mathcal{G}_{\mathrm{ER}}(\nn,p)$ with $1/\nn \leq p\leq 1-1/\nn$.
Under the same conditions in Corollary~\ref{coro:1}, there exists an algorithm that guarantees $\lim_{\nn\rightarrow\infty}\poe=0$ using $\ns=O\left(\nn h(p)\right)$ measurements. Conversely, the number of measurements must satisfy 
$\ns=\Omega\left(\nn h(p)\right)$
to make the probability of error $\poe$ less than ${1}/{2}$.
\end{corollary}

\begin{proof}
Taking ${K}=\nn h(p)/\log\nn$ and ${\mu}={2\nn h(p)}/{(\ln 1/p)}$, we check that ${\mu}<\nn^{-3/{\mu}}(\nn-{K})$ and ${K}=o(\nn)$. The assumptions on $h(p)$ guarantee that $h(p)\geq\log\nn/\nn$, whence $\nn h(p)=\omega\left(\log(\nn/{K})\right)$. The choices of $\{\mu(\nn)\}$ and $\{K(\nn)\}$ make sure that the sequence of distributions is $(\mu(\nn),K(\nn))$-sparse. Theorem~\ref{thm:4} implies that $\ns=O(\nn h(p))$ is sufficient for achieving a vanishing probability of error. For the second part of the corollary, substituting $\mathbbm{H}(\mathcal{G}_{\mathrm{ER}}(\nn,p))=h(p){\nn\choose 2}=\Omega\left(\nn^2 h(p)\right)$ into (\ref{4.20}) yields the desired result.
\end{proof}

\subsection{Measurements corrupted by AWGN}
\label{sec:noisy}
The results on sample complexity can be extended to the case with noisy measurements. The following theorem is proved by combining Theorem~\ref{thm:3} and Lemma~\ref{lemma:2}. The details can be found in Appendix~\ref{app:proof_noisy}.

\begin{theorem}[Noisy worst-case sample complexity]
\label{thm:noisy_sample_complexity}
Suppose that $\mathbf{B}$ and $\mathbf{Z}$ are defined as in Lemma~\ref{lemma:2}. Let ${\mu}<\nn^{-3/{\mu}}(\nn-{K})$ and ${K}=o(\nn)$. Conversely, there exists a $({\mu},{K})$-sparse sequence of distributions such that the number of measurements must satisfy 
$$\ns=\Omega\left(\frac{{\mu}\log({\nn}/{{\mu}})+{K}/\nn^{3/{\mu}}}{\log(1+\sigma^2_{\mathrm{S}}/\sigma^2_{\mathrm{N}})}\right)$$
to make the probability of error $\poet$ less than ${1}/{2}$ for all $\nn$. Moreover, if $\sigma_{\mathrm{N}}=o(1/\nn^{5/2})$, $\sigma_{\mathrm{S}} = 1/\sqrt{\ns}$ and $\nk\leq\mu$,
then
for any sequence of distributions that is $({\mu},{K})$-sparse, the three-stage scheme guarantees that $\lim_{\nn\rightarrow\infty}\poet=0$ using $\ns=O\left({\mu}\log({\nn}/{{\mu}})\right)$ measurements. Moreover, $\lim_{\nn\rightarrow\infty}\poe(\eta)=0$ with $\eta=o(1)$.
\end{theorem}

\begin{remark}
\label{remark:choice_of_gamma}
The proof of Theorem~\ref{thm:noisy_sample_complexity} implies that $\eta=O(\nn^{5/2}\sigma_{\mathrm{N}})$. Therefore, if we consider the normalized Frobenius norm of $(1/n^2)||\mathbf{Y}-\mathbf{X}||_{\mathrm{F}}$ where $\mathbf{X}$ and $\mathbf{Y}$ are the recovered and original graph matrices respectively, then $\sigma_{\mathrm{N}}=o(1/\sqrt{\nn})$ guarantees that the normalized Frobenius norm vanishes. For topology identification, we need to consider the Frobenius norm bound, $\eta$, to rule out the worst-case situation and the sufficient condition becomes $\sigma_{\mathrm{N}}=o(1/\nn^{5/2})$. Another implication is that the choice of $\gamma$ in (\ref{4.0}) satisfying $\gamma=O(\sqrt{\nn}\sigma_{\mathrm{N}})$ (used in the proof) guarantees the reconstruction criteria and its effectiveness is also validated in our experiments in Section~\ref{sec:sim_noisy}.
\end{remark}

\section{Heuristic Algorithm}
\label{sec:6}
We present in this section an algorithm motivated by the consistency-checking step in the proof of achievability (see Section~\ref{sec:ach}). Instead of checking the consistency of each subset of $\mathcal{V}$ consisting of $\nn-{K}$ nodes, as the three-stage scheme does and which requires $O(\nn^{{K}})$ operations, we compute an estimate $X_j$ for each column of the graph matrix independently and then assign a score to each column based on its symemtric consistency with respect to the other columns in the matrix. The lower the score, the closer the estimate of the matrix column $X_j$ is to the ground truth $Y_j$. Using a scoring function we rank the columns, select a subset of them to be ``correct", and then eliminate this subset from the system. The size of the subset determines the number of iterations. Heuristically, this procedure results in a polynomial-time algorithm to compute an estimate $\mathbf{X}$ of the graph matrix $\mathbf{Y}$. 

The algorithm proceeds in four steps.
 \begin{figure}[h!]
	\centering
	\includegraphics[scale=0.4]{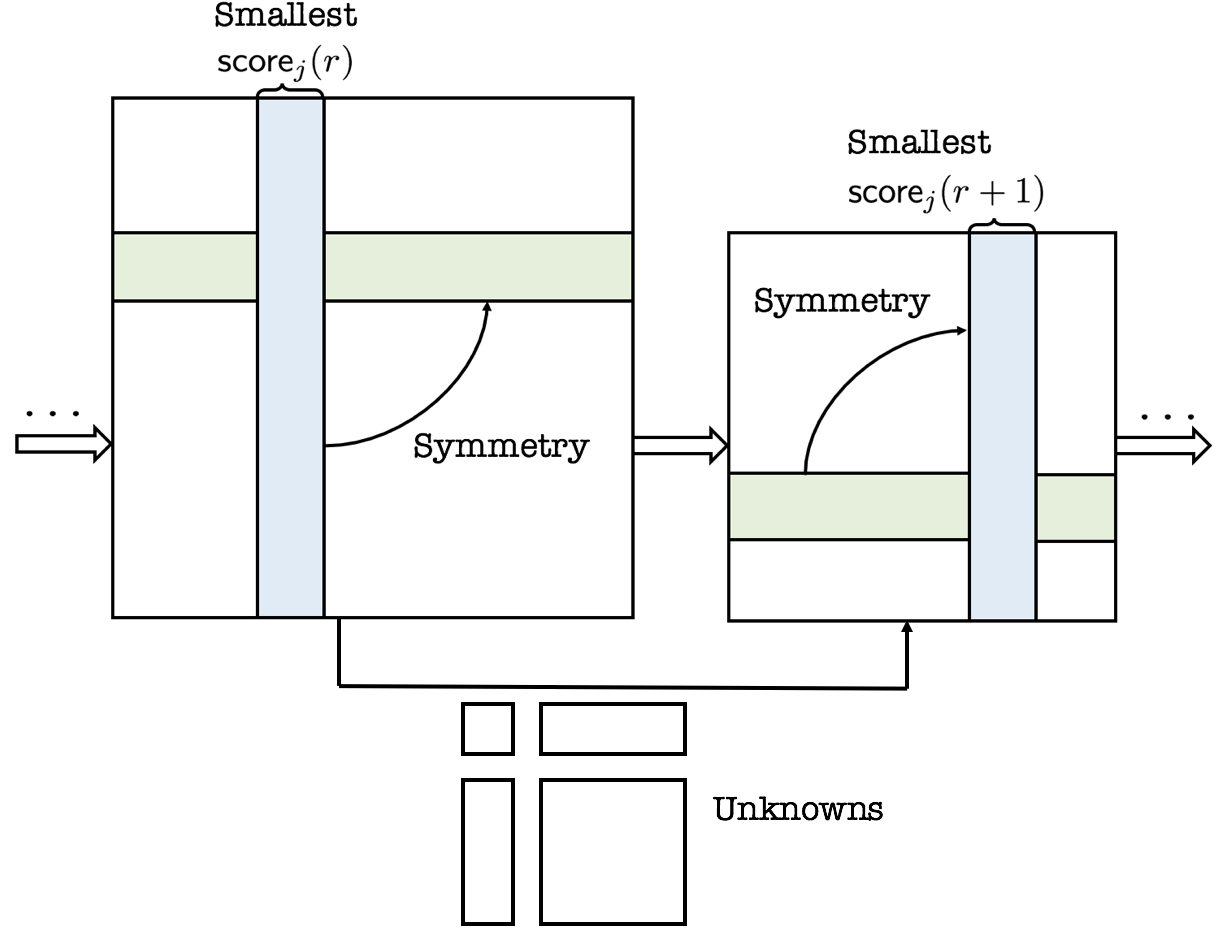}
	\caption{Iterative dimension reduction of the heuristic algorithm. At step $r$, the $s$ columns with the smallest scores defined in~\eqref{eq:score} are assumed to be ``correct" and eliminated from the linear system. The dimension of variables is reduced by $s$ and this procedure is repeated until the $\lceil\nn/s\rceil$ iterations are complete.}
	\label{fig:concept_2}
	\medskip
\end{figure}
\subsubsection{{Step $1$}. Initialization}
Let matrices $\mathbf{A} \in \mathbbm{R}^{\ns \times \nn}$ and $\mathbf{B}\in \mathbbm{R}^{\ns \times \nn}$ be given and set the number of columns fixed in each iteration to be an integer $s$ such that $1 \leq s \leq \nn$. 
For the first iteration, set $\mathcal{S}(0) \leftarrow \mathcal{V}$,
$\mathbf{A}(0)\leftarrow\mathbf{A}$, and
$\mathbf{B}(0)\leftarrow\mathbf{B}$.

For each iteration $r=0,\ldots,\lceil {\nn}/{s}\rceil-1$, we perform the remaining three stages. The system dimension is reduced by $s$ after each iteration.
\subsubsection{{Step $2$}. Independent $\ell_1$-minimization}

For all $j\in\mathcal{S}(r)$, we solve the following $\ell_1$-minimization:

\begin{align}
\label{opt:xr}
X_j(r) = \argmin_{x \in \mathbbm{F}^{\nn-sr}} \quad &{\big|\big|x\big|\big|_{1}}\\
\label{eq:constraint}
\mathrm{subject } \ \mathrm{ to}\quad  & ||\mathbf{B}(r)x - A_j(r)||_{2}\leq \gamma,\\
\nonumber
&x\in \mathcal{X}_j(r).
\end{align}	
Constraint (\ref{opt:xr}) is optional; the set $\mathcal{X}_j(r)$ may encode additional constraints on the form of $x$ such as entry-wise positivity or negativity (\textit{e.g.,} Section \ref{sec:sim}). The forms of reduced matrix $\mathbf{B}(r)$ and reduced vector $A_j(r)$ are specified in Step 4.
\subsubsection{{Step $3$}. Column scoring}
We rank the \textit{symmetric consistency} of the independently solved columns. For all $j\in\mathcal{S}(r)$, let
\begin{align}
\label{eq:score}
\mathsf{score}_j(r):= \sum_{i=1}^{\nn-sr}\left|X_{i,j}(r) - X_{j,i}(r) \right|.
\end{align}
Note that if $\mathsf{score}_j(r)=0$ then $X_j(r)$ and its partner symmetric row in $\mathbf{X}(r)$ are identical. Otherwise there will be some discrepancies between the entries and the sum will be positive. The subset of the $X_j(r)$ corresponding to the $s$ smallest values of $\mathsf{score}_j(r)$ is deemed ``correct". Call this subset of correct indices $\mathcal{S}'(r)$. 

\subsubsection{\text{{Step $4$}. System dimension reduction}}
Based on the assumption that $s$ of the previously computed columns $X_j(r)$ are correct, the dimension of the linear system is reduced by $s$. We set $\mathcal{S}(r+1)\leftarrow \mathcal{S}(r)\backslash\mathcal{S}'(r)$. For all $i,j\in\mathcal{S}'(r)$, we fix
\begin{equation}
\label{eq:recover1}
X_{i,j} =	X_{i,j}(r), \ X_{j,i}= X_{i,j}(r).
\end{equation}
The measurement matrices are reduced to
\begin{align*}
\mathbf{B}(r+1)&\leftarrow \underline{\mathbf{B}}_{\mathcal{S}(r+1)},\\
			\nonumber
			A_j(r+1)&\leftarrow \underline{A}_j(r)-\sum_{i\in\mathcal{S}'(r)}\underline{B}_i X_{i,j}.
\end{align*}
When $r\leq \nn-\ns$, $\underline{\mathbf{B}}_{\mathcal{S}(r+1)}=\mathbf{B}_{\mathcal{S}(r+1)}$, $\underline{A}_j(r)={A}_j(r)$ and $\underline{B}_i={B}_i$.
When $r>\nn-\ns$, to avoid making the reduced matrix $\mathbf{B}(r+1)$ over-determined, we set $\mathbf{B}(r+1)$ to be an $(\nn-r)\times (\nn-r)$ sub-matrix of $\mathbf{B}_{\mathcal{S}(r+1)}$ by selecting $\nn-r$ rows of $\mathbf{B}_{\mathcal{S}(r+1)}$ uniformly at random. A new length-$(\nn-r)$ vector $\underline{A}_j(r)$ is formed by selecting the corresponding entries from $A_j(r)$.
Once the $\lceil {\nn}/{s}\rceil$ iterations complete, an estimate  $\mathbf{X}$ is returned using (\ref{eq:recover1}). The algorithm requires at most $\lceil {\nn}/{s}\rceil$ iterations and in each iteration, the algorithm solves an $\ell_1$-minimization and updates a linear system. Solving an $\ell_1$-minimization can be done in polynomial time (\textit{c.f.}~\cite{ge2011note}). Thus, the heuristic algorithm is a polynomial-time algorithm. 

\section{Applications in Electric Grids}
\label{sec:sim}
Experimental results for the heuristic algorithm are given here for both synthetic data and IEEE standard power system test cases. The algorithm was implemented in Matlab; simulated power flow data was generated using Matpower 7.0~\cite{zimmerman2011matpower} and CVX 2.1~\cite{cvx} with the Gurobi solver~\cite{gurobi} was used to solve the sparse optimization subroutine.  
 
\subsection{Scalable Topologies and Error Criteria}
\label{sec:7.a}
We first demonstrate our results using synthetic data and two typical graph ensembles -- stars and chains. For both topologies, we increment the graph size from $\nn=5$ to $\nn=300$ and record the number of samples required for accurate recovery of parameters and topology. For each simulation, we generate a complex-valued random admittance matrix $\mathbf{Y}$ as the ground truth. Both the real and imaginary parts of the line impedances of the network are selected uniformly and IID from $[-100,100]$. A valid electrical admittance matrix is then constructed using these impedances. The real components of the entries of $\mathbf{B}$ are distributed IID according to $\mathcal{V}\left(1, 1\right)$ and the imaginary components according to $\mathcal{V}\left(0, 1\right)$. $\mathbf{A} = \mathbf{Y}\mathbf{B}$ gives the corresponding complex-valued measurement matrix. The parameter $\gamma$ in (\ref{eq:constraint}) is $0$ since we consider noiseless reconstruction here. 

Given data matrices $\mathbf{A},\mathbf{B}$ the algorithm returns an estimate $\mathbf{X}$ of the ground truth $\mathbf{Y}$. We set $s=\left \lceil{n/2} \right \rceil$ for each graph. If an entry of $\mathbf{X}$ has magnitude $|X_{i,j}| <10^{-5}$, then we fix it to be 0. Following this, if $\mathrm{supp}\left(\mathbf{X}\right)=\mathrm{supp}\left(\mathbf{Y}\right)$ then the topology identification is deemed exact. The criterion for accurate parameter reconstruction is $||\mathbf{Y}-\mathbf{X}||_\mathrm{F}/n^2<10^{-6}$. The number of samples $\ns$ (averaged over repeated trials) required to meet both of these criteria is designated as the sample complexity for accurate recovery. The sample complexity trade-off displayed in Figure~\ref{fig:ieee} shows approximately logarithmic dependence on graph size $\nn$ for both ensembles.



\subsection{IEEE Test Cases}
\begin{figure}[h!]
	\centering
	\includegraphics[scale=0.37]{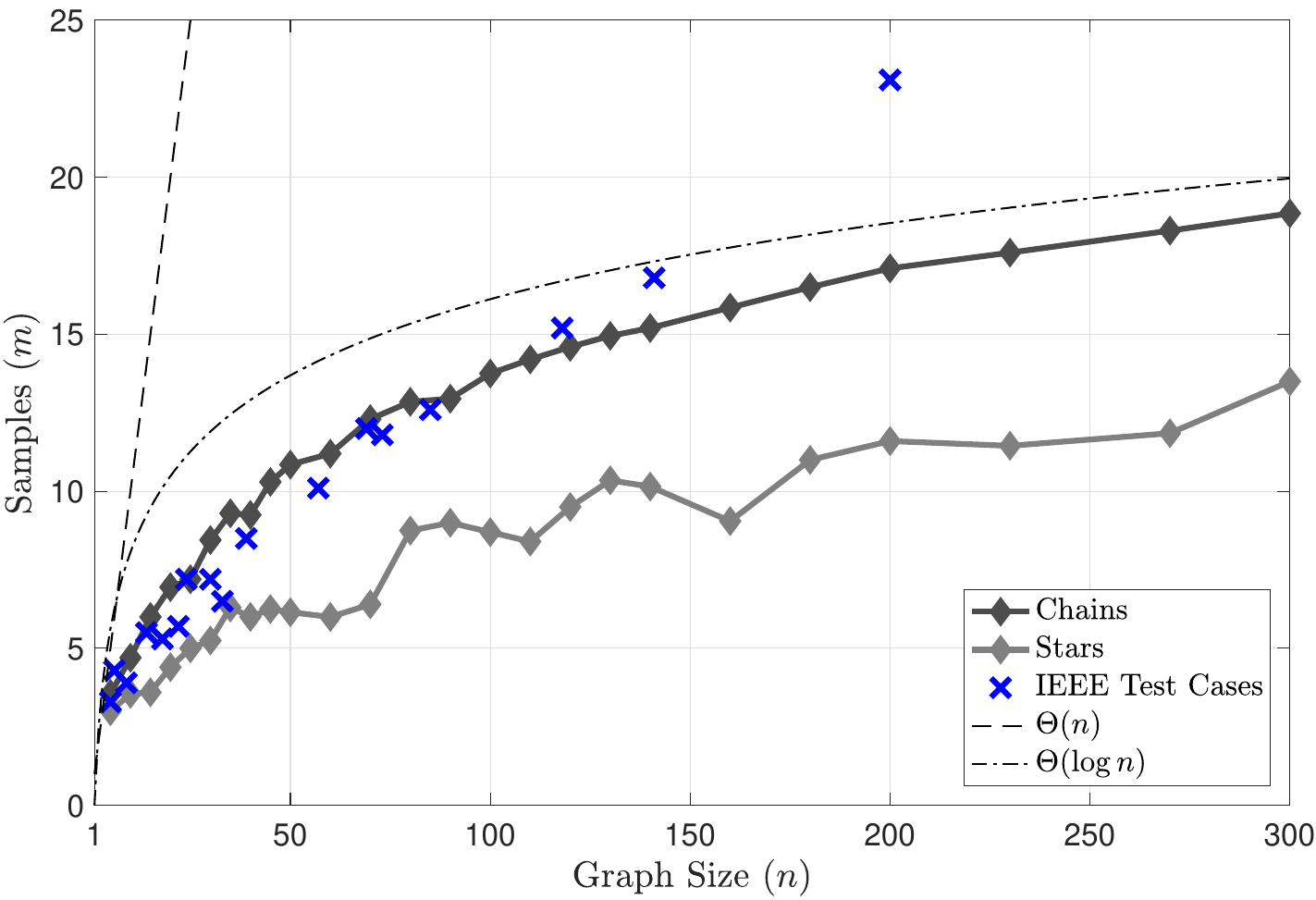}
	\caption{The number of samples required to accurately recover the nodal admittance matrix is shown on the vertical axis. Results are averaged over 20 independent simulations. Star and chain graphs are scaled in size between 5 and 300 nodes. IEEE test cases ranged from 5 to 200 buses. In the latter case, there are no assumptions on the random IID selection of the entries of $\mathbf{Y}$ (in contrast to the star/chain networks). Linear and logarithmic (in $\nn$) reference curves are plotted as dashed lines.}
	\label{fig:ieee}
	\medskip
\end{figure}
We also validate the heuristic algorithm on 17 IEEE standard power system test cases ranging from $5$ to $200$ buses. The procedure for determining sample complexity for accurate recovery is the same as above, but the data generation is more involved.  

\subsubsection{Power flow data generation}
A sequence of time-varying loads is created by scaling the nominal load values in the test cases by a times series of Bonneville Power Administration's aggregate load on $02/08/2016$, 6am to 12pm~\cite{bpa}. For each test case network, we perform the following steps to generate a set of measurements:
\begin{enumerate}[label=\alph*)]
    \item Interpolate the aggregate load profile to $6$-second intervals, extract a length-$m$ random consecutive subsequence, and then scale the real parts of bus power injections by the load factors in the subsequence. 
    \item Compute optimal power flow in Matpower for the network at each time step to determine bus voltage phasors.
    \item Add a small amount of Gaussian random noise ($\sigma^2=0.001$) to the voltage measurements and generate corresponding current phasor measurements using the known admittance matrix.
\end{enumerate}

\subsubsection{Sample complexity for recovery of IEEE test cases}

Figure \ref{fig:ieee} shows the sample complexity for accurate recovery of the IEEE test cases. The procedure and criteria for determining the necessary number of samples for accurate recovery of the admittance matrix are the same as for the synthetic data case. Unlike the previous setting, here we have no prior assumptions about the structure of the IEEE networks: networks have both mesh and radial topologies. However, because power system topologies are typically highly sparse, the heuristic algorithm was able to achieve accurate recovery with a comparable (logarithmic) dependence on graph size.
\begin{figure}[h!]
	\centering
	\includegraphics[scale=0.338]{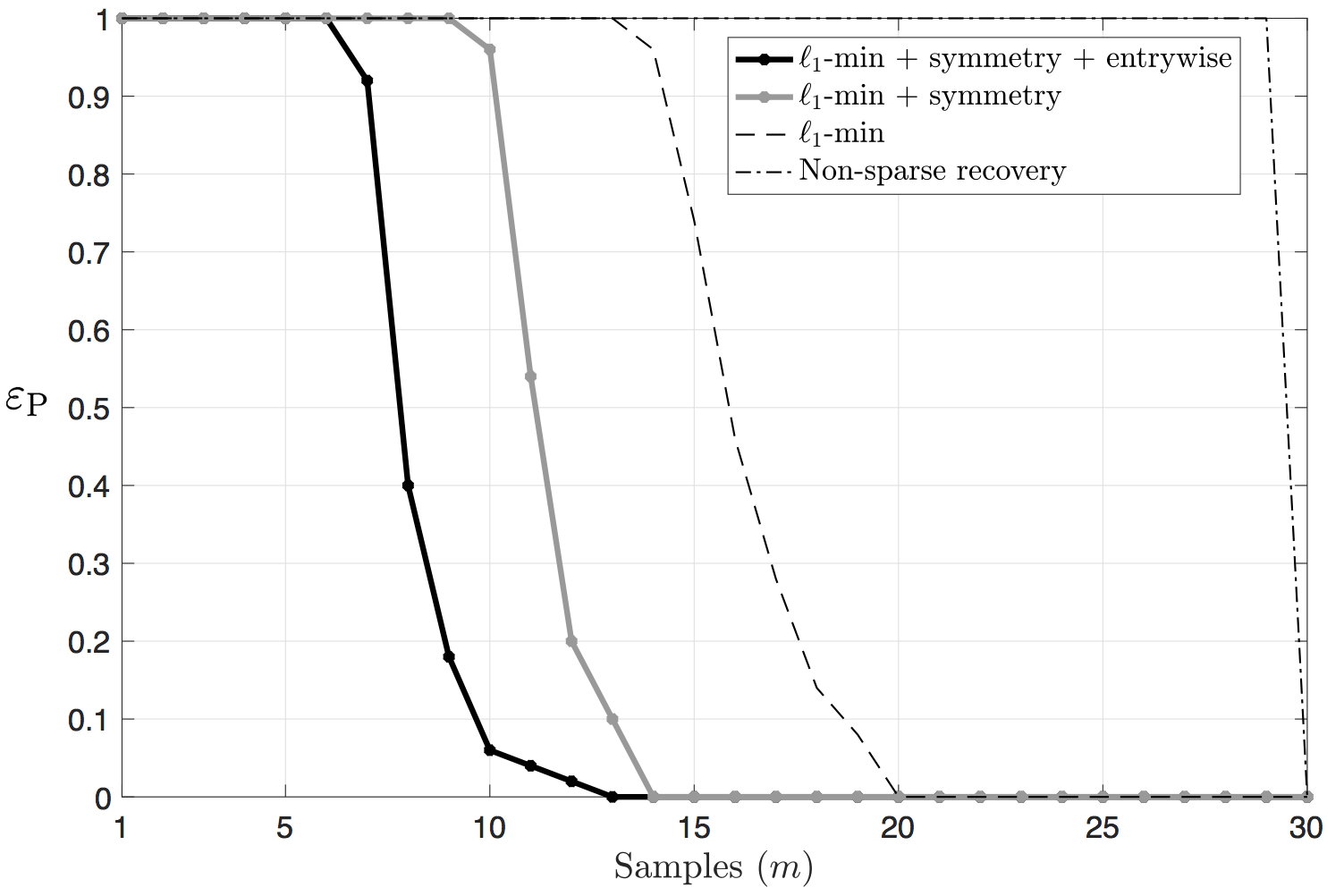}
	\caption{Probability of error for parameter reconstruction $\poe$ for the IEEE $30$-bus test case is displayed on the vertical axis. Probability is taken over 50 independent trials. The horizontal axis shows the number of samples used to compute the estimate $\mathbf{X}$. The probability of error for independent recovery of all $X_j$ via $\ell_1$-norm minimization (double dashed line) and full rank non-sparse recovery (dot dashed line) are shown for reference. Adding the symmetry score function (second-to-left) improves over the naive column-wise scheme. Adding entry-wise positivity/negativity constraints on the entries of $\mathbf{X}$ (left-most curve) reduces sample complexity even further ($\approx 1/3$ samples needed compared to full rank recovery).}
	\label{fig:ieee30}
	\medskip
\end{figure}

\subsubsection{Influence of structure constraints on recovery}
There are structural properties of the nodal admittance matrix for power systems---symmetry, sparsity, and entry-wise positivity/negativity---that we exploit in the heuristic algorithm to improve sample complexity for accurate recovery. The score function $\mathsf{score}_j(r)$ rewards symmetric consistency between columns in $\mathbf{X}$; the use of $\ell_1$-minimization promotes sparsity in the recovered columns; and the constraint set $\mathcal{X}_j$ in (\ref{opt:xr}) forces $\mathrm{Re}(X_{i,j}) \leq 0,\ \mathrm{Im}(X_{i,j}) \geq 0$ for $i \neq j$ and $\mathrm{Re}(X_{i,j}) \geq 0$ for $i=j$. These entry-wise properties are commonly found in power system admittance matrices. In Figure \ref{fig:ieee30} we show the results of an experiment on the IEEE $30$-bus test case that quantify the effects of the structure constraints on the probability of error. In Figure \ref{fig:Gaussian} we show that the score function and the constraints are effective across a range of IEEE test cases, compared with the standard compressed sensing recovery discussed in Section~\ref{sec:lit_cs}. Furthermore, this demonstrates the heuristic algorithm is robust to noise for a broad range of real-world graph structures with respect to Frobenius norm error.

\begin{figure}[h!]
	\centering
	\includegraphics[scale=0.37]{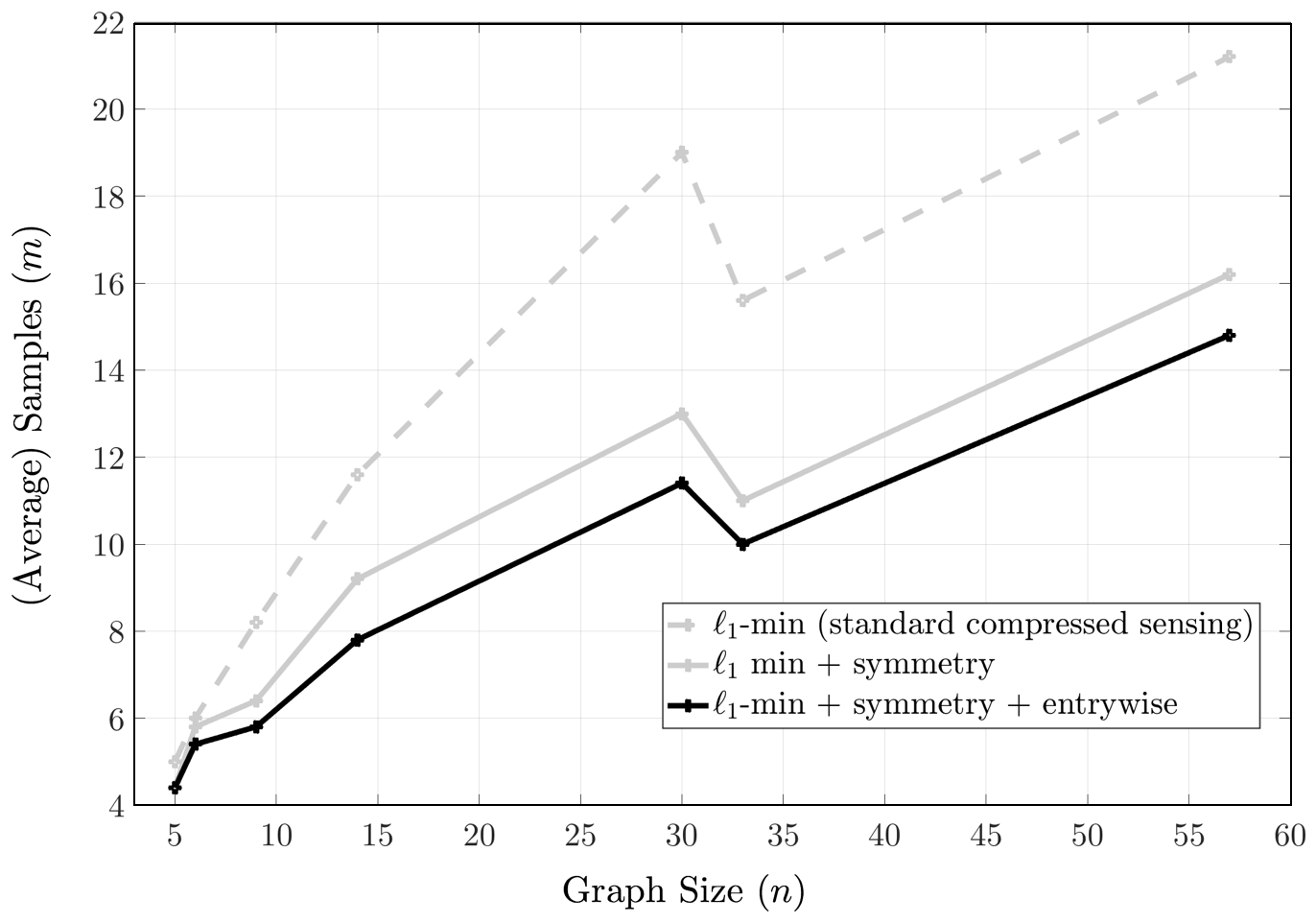}
	\caption{Sample complexity for accurate recovery is shown for a selection of IEEE power system test cases ranging from 5 to 57 buses. The number of samples for accurate recovery is obtained by satisfying the criterion $||\mathbf{X}-\mathbf{Y}||_{\mathrm{F}}/\nn^2 < 10^{-4}$. The noise $\mathbf{Z}$ is an IID Gaussian matrix with zero mean and standard deviation $0.01$. The parameter $\gamma$ in~(\ref{eq:constraint}) is set to be $10^{-4}$. As a benchmark, the number of measurements required for separately reconstructing every column of $\mathbf{Y}$ (standard compressed sensing) is also given.}
	\label{fig:Gaussian}
	\medskip
\end{figure} 

\subsubsection{Comparison with basis pursuit on star graphs}
\label{sec:basis_pursuit}

\begin{figure}[h!]
	\centering
	\includegraphics[scale=0.4]{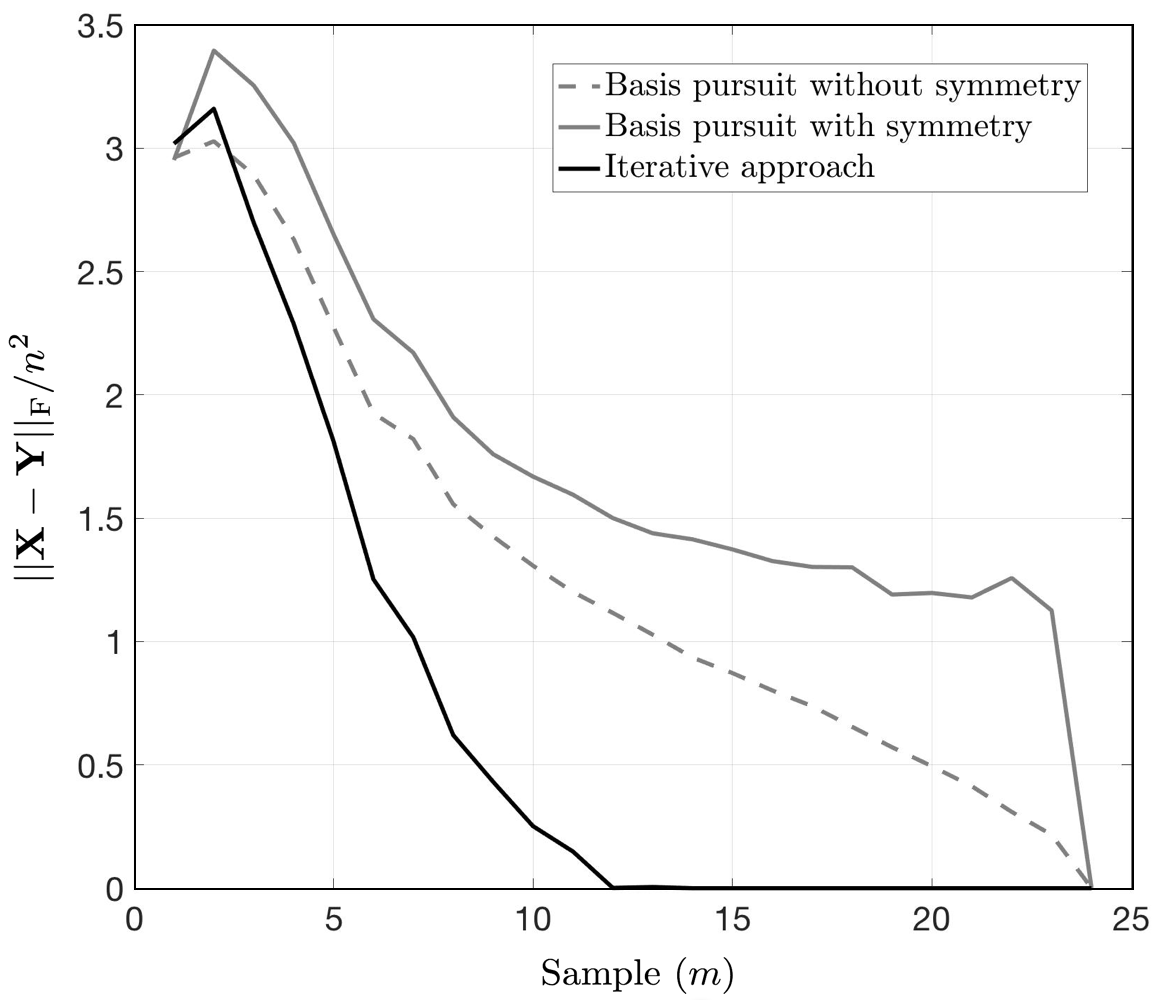}
	\caption{A comparison between our iterative heuristic and basis pursuit. The Frobenius norm error plotted is averaged over $250$ independent trials. The underlying graph is a star graph with $\nn=24$. The solid and
	dotted gray curves are results for basis pursuit with and without a constraint emphasizing symmetry, respectively.}
	\label{fig:comparison}
	\medskip
\end{figure} 

In Figure~\ref{fig:comparison}, we consider star graphs and compare our heuristic algorithm with the modified basis pursuit subroutine in (\ref{eq:basis_pursuit_1})-(\ref{eq:basis_pursuit_2}) with noiseless measurements. For a star graph with $\nn = 24$ nodes, the iterative recovery scheme with $s=12$ outperforms the basis pursuit, with or without a symmetry constraint. The solid and dotted gray curves show the normalized Frobenius error for cases where $\mathbf{Y}(G)$ is constrained to be symmetric and where it is not, respectively.  Our experiments show that convex optimization-based approach breaks down if there are highly dense columns in $\mathbf{Y}$. The star graph contains a high-degree node (degree $\nn-1$), hindering the standard compressed sensing (basis pursuit without the symmetry constraint) from recovering the whole matrix until the number of measurements reaches $\nn$. Surprisingly, adding the symmetry constraint suggests basis pursuit performs less well than basis pursuit without the symmetry condition. This is evidence to support the assumption made in~\cite{dasarathy2015sketching}. There, the non-zeros in the matrix to be recovered should not be
concentrated in any single column (or row) of $\mathbf{Y}(G)$.


\subsubsection{Effects of noise and selection of $\gamma$}
\label{sec:sim_noisy}
In this section, we consider noisy measurements and fix the additive noise $\mathbf{Z}$ be IID Gaussian with mean zero and  variance $\sigma^2_{\mathrm{N}}\in [10^{-9},10^{-2}]$. We set  $\gamma=\sqrt{\nn}\sigma_{\mathrm{N}}$ in (\ref{eq:constraint}), as indicated in Remark~\ref{remark:choice_of_gamma}. Due to the presence of noise, there is error in the recovered matrix $\mathbf{X}$. However, the mean absolute percentage error is small.  
\begin{figure}[h!]
	\centering
	\includegraphics[scale=0.35]{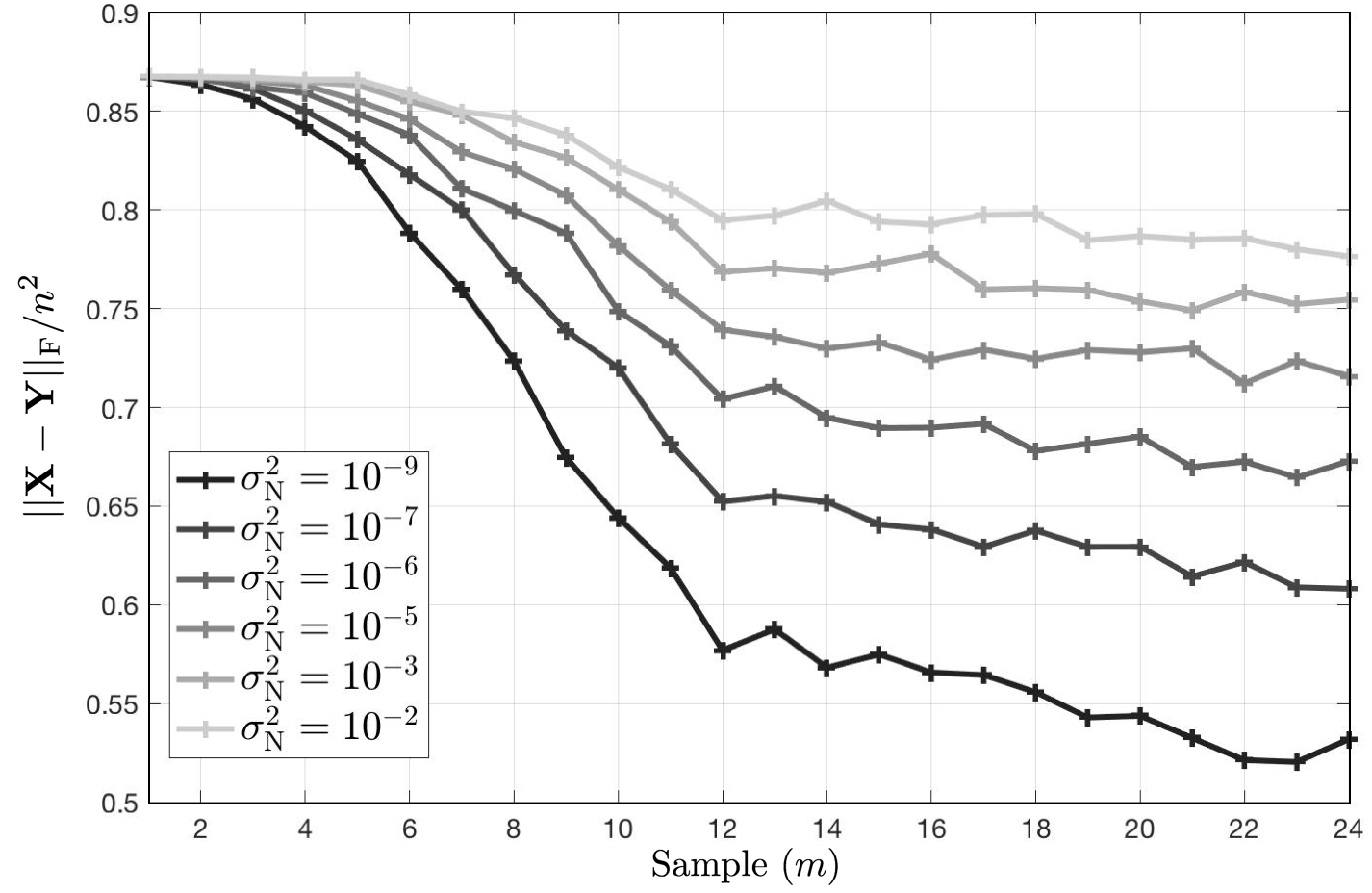}
	\caption{
	The impact of measurement noise on sample complexity for recovery of the IEEE 24-bus RTS test case is demonstrated. Trajectories correspond to increasing noise levels from dark (least) to light (most). From left to right, we observe---as expected---that for each variance value, the normalized Frobenius error of the recovered matrix decreases as the number of samples used for recovery increases. From bottom to top, we observe that the error increases (for every value of $\ns$) as variance of the additive noise $\mathbf{Z}$ increases.
	}
	\label{fig:noisy}
	\medskip
\end{figure}

\addcontentsline{toc}{section}{Bibliography}
{\bibliographystyle{IEEEtran}
	\small
	\bibliography{ref}}

\appendix

\section{Erratum}

In the proof of Lemma 1 in~\cite{li2020learning}, the KL divergence inequality $\mathbbm{D}_{\mathrm{KL}}\left(\frac{{\mu}}{\nn-2}\big|\big| \frac{1}{\nn}\right)\leq\frac{{\mu}}{\nn-2}\ln\nn$ was applied in its reverse direction. The lemma needs to be fixed as follows. The lemma has been amended accordingly. We extend our gratitude to Carl Wanninger for identifying the error.
\begin{lemma}[Original Version]
\label{lemma:trees}
For any ${\mu}\geq 1$ and ${K}>0$, the distribution $\mathcal{U}_{\mathsf{T}(\nn)}$ is $({\mu },{K},1/K)$-{sparse}. 
\end{lemma}
\begin{lemma}[Fixed Version]
\label{lemma:trees}
For any ${\mu}\geq 1+\ln\nn$ and ${K}>0$, the distribution $\mathcal{U}_{\mathsf{T}(\nn)}$ is $({\mu },{K},1/K)$-{sparse}. 
\end{lemma}

\begin{corollary}[Original Version]
\label{coro:1}Let $\mathbf{Z}=0$.
Suppose that the generator matrix $\mathbf{B}$ has Gaussian IID entries with mean zero and variance one and assume $G_\nn$ is distributed according to $\mathcal{U}_{\mathsf{T}(\nn)}$. There exists an algorithm that guarantees $\lim_{\nn\rightarrow\infty}\poe=0$ using $\ns=O\left(\log\nn\right)$ measurements. Conversely, the number of measurements must satisfy 
$\ns=\Omega\left(\log\nn\right)$
to make the probability of error $\poe$ less than ${1}/{2}$.
\end{corollary}

\begin{corollary}[Fixed Version]
\label{coro:1}Let $\mathbf{Z}=0$.
Suppose that the generator matrix $\mathbf{B}$ has Gaussian IID entries with mean zero and variance one and assume $G_\nn$ is distributed according to $\mathcal{U}_{\mathsf{T}(\nn)}$. There exists an algorithm that guarantees $\lim_{\nn\rightarrow\infty}\poe=0$ using $\ns=O\left((\log\nn)^2\right)$ measurements. Conversely, the number of measurements must satisfy 
$\ns=\Omega\left(\log\nn\right)$
to make the probability of error $\poe$ less than ${1}/{2}$.
\end{corollary}

\section{Proof of Theorem~\ref{thm:1}}
\label{app:proof_fundamental_limits}
\begin{proof}
The graph $G$ is chosen from a discrete set $\mathsf{C}(\nn)$ according to some probability distribution ${\mathcal{G}_\nn}$.
Fano's inequality~\cite{fano1961transmission} plays an important role in deriving fundamental limits. We especially focus on its extended version. Similar generalizations appear in many places, \textit{e.g.,}~\cite{aeron2010information,santhanam2012information} and~\cite{ghoshal2016information}. We repeat the lemma here for the sake of completion:
\begin{lemma}[\text{Generalized Fano's inequality}]
	\label{lemma:fano}
Let $G$ be a random graph and let $\mathbf{A}$ and $\mathbf{B}$ be matrices defined in Section~\ref{sec:model} and~\ref{sec:cvm}. Suppose the original graph $G$ is selected from a nonempty candidacy set $\mathsf{C}(\nn)$ according to a probability distribution ${\mathcal{G}_\nn}$. Let $\hat{G}$ denote the estimated graph. Then the conditional probability of error for estimating $G$ from $\mathbf{A}$ given $\mathbf{B}$ is always bounded from below as
\begin{align}
\label{3.2}
\mathbbm{P}\left(\hat{G}\neq G\big| \mathbf{B}\right)\geq 1-\frac{\mathbbm{I}\left(G;\mathbf{A}\big|\mathbf{B}\right)+\ln 2}{\mathbbm{H}\left({\mathcal{G}_\nn}\right)}
\end{align}
where the randomness is over the selections of the original graph $G$ and the estimated graph $\hat{G}$.
\end{lemma}

In (\ref{3.2}), the term $\mathbbm{I}\left(G;\mathbf{B}\big|\mathbf{A}\right)$ denotes the conditional mutual information (base $e$) between $G$ and $\mathbf{B}$ conditioned on $\mathbf{A}$.
Furthermore, the conditional mutual information $\mathbbm{I}\left(G;\mathbf{A}|\mathbf{B}\right)$ is bounded from above by the differential entropies of $\mathbf{A}$ and $\mathbf{B}$. It follows that
\begin{align}
\label{3.8}
\mathbbm{I}\left(G;\mathbf{A}|\mathbf{B}\right)&= \mathbbm{H}\left(\mathbf{A}|\mathbf{B}\right)-\mathbbm{H}\left(\mathbf{A}|G,\mathbf{B}\right)\\
\label{3.9}
&\leq  \mathbbm{H}\left(\mathbf{A}|\mathbf{B}\right)-\mathbbm{H}\left(\mathbf{A}|\mathbf{Y},\mathbf{B}\right)\\
\label{3.10}
&= \mathbbm{H}\left(\mathbf{A}|\mathbf{B}\right)-\mathbbm{H}\left(\mathbf{Z}\right)\\
\label{3.11}
&\leq \mathbbm{H}\left(\mathbf{A}\right)-\mathbbm{H}\left(\mathbf{Z}\right).
\end{align}
%
Here, Eq.~(\ref{3.8}) follows from the definitions of mutual information and differential entropy. Moreover, knowing $\mathbf{Y}$, the graph $G$ can be inferred. Thus, $\mathbbm{H}\left(\mathbf{A}|G,\mathbf{B}\right)\geq \mathbbm{H}\left(\mathbf{A}|\mathbf{Y},\mathbf{B}\right)$ yields (\ref{3.9}). Recalling the linear system in (\ref{eq:linear}), we obtain (\ref{3.10}). Furthermore, (\ref{3.11}) holds since $\mathbbm{H}\left(\mathbf{A}\right)\geq \mathbbm{H}\left(\mathbf{A}|\mathbf{B}\right)$.

Plugging (\ref{3.11}) into (\ref{3.2}),
\begin{align*}
\poet=&\mathbbm{E}_{\mathbf{B}}\left[\mathbbm{P}\left(\hat{G}\neq G\big| \mathbf{B}\right)\right]\\
\geq& 1-\frac{ \mathbbm{H}\left(\mathbf{A}\right)-\mathbbm{H}\left(\mathbf{Z}\right)+\ln 2}{\mathbbm{H}\left({\mathcal{G}_\nn}\right)},
\end{align*}
which yields the desired (\ref{3.0}).
\end{proof}

\section{Proof of Theorem~\ref{thm:3}}
\label{app:proof_achieve}
Conditioning on that no less than $\nn-{K}$ many columns are recovered with respect to the $\Gamma$-probability of error, \textit{i,e.,} for each entry, the absolute value of the difference between the recovered one and the original one is bounded from above by $\gamma$, the union bound ensures the desired bound on the probability of error for noisy parameter reconstruction. It remains to show that the consistency-check in our scheme gives the expression for $\eta$. First, if no less than $\nn-{K}$ many columns are recovered, there must be a subset $\mathcal{S}\subseteq\mathcal{V}$ passing through the consistency-check. Let us consider the vectors that are not $\mu$-sparse. For any such vector $Y^*\in\mathbbm{R}^{\nn}$, denote by $e=Y^*-Y'$ the difference of $Y^*$ and the original vector $Y'$. It follows that $e$ can be decomposed as a summation of a $2\nk$-sparse vector $\overline{e}\in\mathbbm{R}^{\nn}$ and a vector $f\in\mathbbm{R}^{\nn}$ that satisfies $|f_i|\leq 2\Gamma$ for all $i\in\mathcal{V}$. Therefore, the definition of restricted isometry constants ensures the following:
\begin{align*}
    ||e||_2 \leq& ||\overline{e}||_2+||f||_2\\
    \leq&\frac{1}{1-\delta_{2\nk}} ||\mathbf{B}\overline{e}||_2+ 2\nn\Gamma\\
    \leq&\frac{1}{1-\delta_{2\nk}} ||\mathbf{B}e||_2+ \left(2\nn+\frac{2||\mathbf{B}||_2}{1-\delta_{2\nk}} \right)\Gamma
\end{align*}
which can be further bounded by noting that
\begin{align*}
||\mathbf{B}e||_2 =& ||(\mathbf{B}Y^*-A)-(\mathbf{B}Y'-A)||_2\leq 2\gamma
\end{align*}
since both $Y'$ and $Y^*$ satisfy (\ref{4.0}) where $A$ is a column of $\mathbf{A}$. Thus, the consistency-check guarantees that for each $j$ in the set $\mathcal{S}\subseteq\mathcal{V}$ that passes the check,
\begin{align*}
||X_j-Y_j||_2\leq 2\left(\nn+\frac{||\mathbf{B}||_2}{1-\delta_{2\nk}} \right)\Gamma+\frac{2\gamma}{1-\delta_{2\nk}}.
\end{align*}
Consider the reduced linear system in (\ref{4.3}). For each $j$ in the set $\overline{\mathcal{S}}\subseteq\mathcal{V}$, 
\begin{align*}
   ||X^{\overline{\mathcal{S}}}_j-Y^{\overline{\mathcal{S}}}_j||_2 \leq & \left|\left|\left(\mathbf{B}^{\mathcal{K}}_{\overline{\mathcal{S}}}\right)^{-1}\right|\right|_2 \left|\left|\mathbf{B}_{\mathcal{S}}(X_j^{\mathcal{S}}-Y_j^{\mathcal{S}})\right|\right|_2\\
   \leq &\left|\left|\left(\mathbf{B}^{\mathcal{K}}_{\overline{\mathcal{S}}}\right)^{-1}\right|\right|_2 \left|\left|\mathbf{B}_{\mathcal{S}}\right|\right|_2\left|\left|X_j^{\mathcal{S}}-Y_j^{\mathcal{S}}\right|\right|_2.
\end{align*}
Summing up the bounds on the $\ell_2$ norms for each column and considering the worst case of the invertible matrix $\mathbf{B}^{\mathcal{K}}_{\overline{\mathcal{S}}}$, the bound $\eta$ on the Frobenius norm follows by arranging the terms.

\section{{Proof of Corollary~\ref{corollary:noiseless}}}
\label{app:proof_noiseless}
\begin{proof}
Conditioned on $G\in\mathsf{C}(\nn)({\mu},{K})$ and the assumption $\delta_{3{\mu}}+3\delta_{4{\mu}}< 2$, there are no less than $\nn-{K}$ many columns correctly recovered. Therefore, any such set $\mathcal{S}$ with $\left|\mathcal{S}\right|=\nn-{K}$ must contain at least $\nn-2K$ many corresponding indexes of the correctly recovered columns. The consistency-checking verifies that if the collection of an arbitrary set of nodes $\mathcal{S}$ of cardinality $\nn-{K}$ satisfies the symmetry property as the true graph $\mathbf{Y}$ must obey. If the consistency-checking fails, it is necessary that there exist two distinct length-$\nn$ vectors $Y'$ and $Y^*$ in $\mathbbm{F}^{\nn}$ such that $Y^*$ is the minimizer of the $\ell_1$-minimization (\ref{4.1})-(\ref{4.2}) that differs from the correct answer $Y'$, \textit{i.e.,} $Y'\neq Y^*$ where $A=\mathbf{B}Y'$ and 
\begin{align*}
Y^*=\argmin_{Y} & \ \left|\left|Y\right|\right|_{1}\\
\mathrm{subject}\ \mathrm{to} \ &A=\mathbf{B}Y\\
&Y\in\mathbbm{F}^n
\end{align*}
for some $A\in\mathbbm{F}^\ns$ and furthermore, the vectors $Y'$ and $Y^*$ can have at most  $2K$ distinct coordinates,
\begin{align*}
\left|\mathrm{supp}\left(Y'-Y^*\right)\right|\leq 2K.
\end{align*}

However, the constraints $\mathbf{B}Y'=A$ and $\mathbf{B}Y^*=A$ imply that $\mathbf{B}\left(Y'-Y^*\right)=0$, contradicting to $\mathrm{spark}(\mathbf{B})>2K$. Therefore, $\nn-{K}$ many columns can be successfully recovered if the decoded solution passes the consistency-checking. Moreover, since $\mathrm{spark}(\mathbf{B})>2K$ and number of unknown coordinates in each length-${K}$ vector $X_{j}^{\overline{\mathcal{S}}}$ (for $j=1,\ldots,|\overline{\mathcal{S}}|$) to be recovered is ${K}$, the solution of the system~(\ref{4.3}) is guaranteed to be unique. Thus, Algorithm~\ref{alg:1} always recovers the correct columns $Y_1,\ldots,Y_N$ conditioned on $\mathrm{spark}(\mathbf{B})>2K$. It follows that
$\poe\leq 1-\mathbbm{P}_{{\mathcal{G}}}(G\in\mathsf{C}(\nn,{\mu},{K}))$
provided $\mathrm{spark}(\mathbf{B})>2K$. In agreement with the assumption that the distribution $\mathcal{G}$ is $({\mu},{K},\rho)$-sparse, (\ref{eq:rec}) must be satisfied. Therefore, the probability of error must be less than $\rho$.
\end{proof}

\section{{Proof of Lemma~\ref{lemma:trees}}}
\label{app:proof_trees}

\begin{proof}
Consider the following function
\begin{align*}
F(\mathcal{E}) = \sum_{j=1}^{\nn} f(d_j(G))
\end{align*}
where $d_j(G)$ denotes the degree of the $j$-th node and consider the following indicator function:
\begin{align*}
f(d_j(G)) := \begin{cases}
1 \quad &\text{if } d_j(G)> {\mu}\\
0 & \text{otherwise}
\end{cases}.
\end{align*}

Applying the Markov's inequality,
\begin{align}
\nonumber
\mathbbm{P}\left(G\notin\mathsf{T}(\nn)({\mu},{K})\right)&=\mathbbm{P}_{\mathcal{U}_{\mathsf{T}(\nn)}}\left(F(\mathcal{E})\geq {K}\right)\\
\label{eq:a.1}
&\leq \frac{\mathbbm{E}_{\mathcal{U}_{\mathsf{T}(\nn)}}\left[F(\mathcal{E})\right]}{{K}}.
\end{align}
Continuing from~(\ref{eq:a.1}), the expectation $\mathbbm{E}_{\mathcal{U}_{\mathsf{T}(\nn)}}\left[F(\mathcal{E})\right]$ can be further expressed and bounded as
\begin{align}
\nonumber
\mathbbm{E}_{\mathcal{U}_{\mathsf{T}(\nn)}}\left[F(\mathcal{E})\right] &=
\sum_{j=1}^{\nn}\mathbbm{E}_{\mathcal{U}_{\mathsf{T}(\nn)}}\left[f(d_j(G))\right]\\
\label{eq:a.10}
& = \sum_{j=1}^{\nn} \mathbbm{P}_{\mathcal{U}_{\mathsf{T}(\nn)}}\left(d_j(G)> {\mu}\right).
\end{align}
Since $G$ is chosen uniformly at random from $\mathsf{T}(\nn)$, it is equivalent to selecting its corresponding Pr\"{u}fer sequence (by choosing $\nn-2$ integers independently and uniformly from the set $\mathcal{V}$, \textit{c.f.} ~\cite{kajimoto2003extension}) and the number of appearances of each $j\in\mathcal{V}$ equals to $d_j(G)-1$. Therefore, for any fixed node $j\in\mathcal{V}$, the Chernoff bound implies that 
\begin{align}
\label{eq:a.12}
\mathbbm{P}_{\mathcal{U}_{\mathsf{T}(\nn)}}\left(d_j(G)> {\mu}\right)\leq \exp\left(-(\nn-2) \mathbbm{D}_{\mathrm{KL}}\left(\frac{{\mu}}{\nn-2}\big|\big| \frac{1}{\nn}\right)\right)
\end{align}
where $\mathbbm{D}_{\mathrm{KL}}(\cdot || \cdot)$ is the Kullback-Leibler divergence of Bernoulli distributions and
\begin{align}
\label{eq:a.11}
\mathbbm{D}_{\mathrm{KL}}\left(\frac{{\mu}}{\nn-2}\big|\big| \frac{1}{\nn}\right)\geq 3\frac{\left(\frac{\nn\mu-(\nn-2)}{\nn(\nn-2)}\right)^2}{\frac{2}{\nn}+\frac{\mu}{\nn-2}}\geq 3\frac{(\mu-1)^2}{(\nn-2)(\mu + 2)}.
\end{align}

Therefore, substituting (\ref{eq:a.11}) back into (\ref{eq:a.12}) and combining (\ref{eq:a.1}) and (\ref{eq:a.10}), setting ${\mu}\geq 1 + \ln \nn$ leads to
\begin{align*}
\mathbbm{P}\left(G\notin\mathsf{T}(\nn)({\mu},{K})\right)\leq \frac{\nn}{K} \exp\left(-3\frac{({\mu}-1)^2}{\mu+2}\right) \leq \frac{1}{K}.
\end{align*}
\end{proof}

\section{{Proof of Lemma~\ref{lemma:renyi}}}
\label{app:proof_renyi}
\begin{proof}
For any fixed node $j\in\mathcal{V}$, applying the Chernoff bound,
\begin{align*}
\mathbbm{P}_{\mathcal{G}_{\mathrm{ER}}(\nn,p)}\left(d_j(G)> {\mu}\right)\leq \exp\left(-\nn \mathbbm{D}_{\mathrm{KL}}\left(\frac{{\mu}}{\nn}\big|\big| p\right)\right).
\end{align*}

Continuing from~(\ref{eq:a.1}), the expectation $\mathbbm{E}_{\mathcal{G}_{\mathrm{ER}}(\nn,p}\left[F(\mathcal{E})\right]$ can be further expressed and bounded as
\begin{align}
\mathbbm{E}_{\mathcal{G}_{\mathrm{ER}}(\nn,p)}\left[F(\mathcal{E})\right] 
\label{eq:a.3}
& \leq \nn \cdot \exp\left(-\nn \mathbbm{D}_{\mathrm{KL}}\left(\frac{{\mu}}{\nn}\big|\big| p\right)\right)
\end{align}
where the probability $p$ satisfies $0<p\leq {\mu}/\nn<1$. Note that
\begin{align}
\label{eq:a.2}
\mathbbm{D}_{\mathrm{KL}}\left(\frac{{\mu}}{\nn}\big|\big| p\right)= \frac{{\mu}}{\nn}\ln \frac{1}{p}+\left(1-\frac{{\mu}}{\nn}\right)\ln \frac{1}{1-p} - h(p)
\end{align}
where the binary entropy $h(p)$ is in base $e$.
Taking ${\mu}\geq{2\nn h(p)}/{(\ln 1/p)}\geq 2\nn p$, substituting (\ref{eq:a.2}) into (\ref{eq:a.3}) leads to
\begin{align*}
\mathbbm{E}_{\mathcal{G}_{\mathrm{ER}}(\nn,p)}\left[F(\mathcal{E})\right] \leq \nn\exp\left(-\nn h(p)\right).
\end{align*}
Therefore, (\ref{eq:a.1}) gives 
\begin{align*}
\mathbbm{P}\left(G\notin\mathsf{C}(\nn)({\mu},{K})\right)\leq \frac{\nn\exp\left(-\nn h(p)\right)}{{K}}.
\end{align*}
\end{proof}

\section{Proof of Lemma~\ref{lemma:2}}
\label{app:proof_of_Gaussian_converse}
\begin{proof}
Continuing from Theorem~\ref{thm:1}, 
\begin{align}
\nonumber
&\mathbbm{H}(\mathbf{A}) - \mathbbm{H}(\mathbf{Z})\\
\nonumber
=& \sum_{i=1}^{\ns}\left[\mathbbm{H}\left({A}^{(i)}\right) - \mathbbm{H}\left({Z}^{(i)}\right)\right]\\
\label{eq:5.1}
\overset{\mathrm{(a)}}{\leq}&
\sum_{i=1}^{\ns}\frac{\nn}{2}\left[\ln\left(2\pi e \frac{\mathrm{Tr}\left(\Sigma_{\mathbf{A^{(i)}}}\right)}{\nn}\right)-\ln(2\pi e \sigma^2_{\mathrm{N}})\right]
\end{align}
where $\mathrm{Tr}\left(\Sigma_{\mathbf{A^{(i)}}}\right)$ is the trace of the covariance matrix of $\mathbf{A^{(i)}}$ and we have used the fact that normal distributions maximize entropy and the inequality $\mathrm{det}(\Sigma_{\mathbf{A^{(i)}}})\leq (\mathrm{Tr}\left(\Sigma_{\mathbf{A^{(i)}}}\right)/\nn)^{\nn}$ to obtain (a). Note that because of the assumption of independence, the trace is bounded from above by $\nn\sigma^2_{\mathrm{S}}\overline{Y}+\nn\sigma^2_{\mathrm{N}}$ where $\overline{Y}:=\max_{i,j}|Y_{i,j}|$. Substituting this into (\ref{eq:5.1}) completes the proof. The special case when $\mathbf{Z}=0$ follows similarly.
\end{proof}

\section{{Proof of Theorem~\ref{thm:4}}}
\label{app:proof_sample}
\begin{proof}
The first part is based on Corollary~\ref{corollary:noiseless}. Under the assumption of the generator matrix $\mathbf{B}$, using Gordon's escape-through-the-mesh theorem, Theorem $4.3$ in~\cite{rudelson2008sparse} implies that for any columns $Y_j$ with $j\in\mathcal{V}_{\mathrm{Small}}$ are correctly recovered using the minimization in (\ref{4.1})-(\ref{4.2}) with probability at least $1-2.5\exp\left(-(4/9){\mu}\log(\nn/{\mu})\right)$, as long as the number of measurements satisfies $\ns\geq 48{\mu}\left(3+2\log(\nn/{\mu})\right)$, and $\nn/{\mu}>2, {\mu}\geq 4$ (if ${\mu}\leq 3$, the multiplicative constant increases but our theorem still holds). Similar results were first proved by Candes, \textit{et al.} in~\cite{candes2005error} (see their Theorem $1.3$). Therefore, applying the union bound, the probability that all the ${\mu}$-sparse columns can be recovered simultaneously is at least $1-2.5\nn\exp\left(-(4/9){\mu}\log(\nn/{\mu})\right)$. On the other hand, conditioned on that all the ${\mu}$-sparse columns are recovered, Corollary~\ref{corollary:noiseless} indicates that $\mathrm{spark}(\mathbf{B})>2K$ is sufficient for the three-stage scheme to succeed. Since each entry in $\mathbf{B}$ is an IID Gaussian random variable with zero mean and variance one, if $\ns\geq 48{\mu}\left(3+2\log(\nn/{\mu})\right)+2K$, with probability one that the spark of $\mathbf{B}$ is greater than $2K$, verifying the statement.


The converse follows by applying Lemma~\ref{lemma:2} with $\mathbf{Z}=0$. Consider the uniform distribution $\mathcal{U}_{\mathsf{C}(\nn)({\mu},{K})}$ on $\mathsf{C}(\nn)({\mu},{K})$.
Then $\mathbbm{H}\left(\mathcal{U}_{\mathsf{C}(\nn)({\mu},{K})}\right)=\ln\left|\mathsf{C}(\nn)({\mu},{K})\right|$. Let $0\leq\alpha,\beta\leq 1$ be parameters such that ${\mu}<\beta(\nn-\alpha {K})$. To bound the size of $\mathsf{C}(\nn)({\mu},{K})$, we partition $\mathcal{V}$ into $\mathcal{V}_{1}$ and $\mathcal{V}_{2}$ with $|\mathcal{V}_{1}|=\nn-\alpha {K}$ and $|\mathcal{V}_{2}|=\alpha {K}$. First, we assume that the nodes in $\mathcal{V}_{1}$ form a ${\mu}/2$-regular graph. For each node in $\mathcal{V}_{2}$, construct $\beta(\nn-\alpha {K})\in\mathbbm{N}_{+}$ edges and connect them to the other nodes in $\mathcal{V}$ with uniform probability. A graph constructed in this way always belongs to $\mathsf{C}(\nn)({\mu},{K})$, unless the added edges create more than ${K}$ nodes with degrees larger than ${\mu}$. Therefore, as $\nn\rightarrow\infty$,
\begin{align}
\label{eq:change1}
\left|\mathsf{C}(\nn)({\mu},{K})\right|\geq&  \rho\cdot\frac{e^{1/4}\displaystyle\binom{\,N-1}{\phi}^{N}{\displaystyle\binom{\binom{N}{2}}{\phi N/2}}}{\displaystyle\binom{N(N-1)}{\phi N}}
\cdot\binom{\nn-1}{M}^{\alpha {K}}
\end{align}
where $N:=\nn-\alpha {K}$, $M:=\beta(\nn-\alpha {K})$ and $\phi:={\mu}/2$. The first term $\rho$ denotes the fraction of the constructed graphs that are in $\mathsf{C}(\nn)({\mu},{K})$. The second term in (\ref{eq:change1}) counts the total number of $\phi$-regular graphs~\cite{liebenau2017asymptotic}, and the last term is the total number of graphs created by adding new edges for the nodes in $\mathcal{V}_2$. If ${K}=O({\mu})$, there exists a constant $\alpha>0$ small enough such that $\rho=1$. If ${\mu}=o({K})$, for any fixed node in $\mathcal{V}_1$, the probability that its degree is larger than ${\mu}$ is
\begin{align*}
&\sum_{i=\phi+1}^{\alpha {K}}\binom{\alpha {K}}{i}\beta^{i}(1-\beta)^{\alpha {K}-i}\\
\leq&
\sum_{i=\phi+1}^{\alpha {K}}\alpha {K} h\left(\frac{i}{\alpha {K}}\right)\beta^i
\leq (\alpha {K})^2 \beta^{\phi+1}
\end{align*}
where $h({i}/{\alpha {K}})$ is in base $e$.
Take $\beta=\nn^{-3/{\mu}}$ and $\alpha=1/2$. The condition ${\mu}<\nn^{-3/{\mu}}(\nn-{K})$ guarantees that ${\mu}<\beta(\nn-\alpha {K})$. Letting $\digamma(\nn):=1/\nn$ be the assignment function for each node in $\mathcal{V}_1$, we check that
\begin{align*}
(\alpha {K})^2 \beta^{\phi+1}\leq \frac{1}{4\nn} \leq \digamma(\nn)\cdot \left(1-\frac{1}{\digamma(\nn)}\right)^{N}\leq\frac{1}{e\nn}.
\end{align*}
Therefore, applying the Lov\'{a}sz local lemma, the probability that all the nodes in $\mathcal{V}_1$ have degree less than or equal to $\mu$ can be bounded from below by $ \left(1-\digamma(\nn)\right)^N \geq 1/4$ if $\nn\geq 2$, which furthermore is a lower bound on $\rho$. Therefore,
taking the logarithm,
\begin{align}
\nonumber
\mathbbm{H}\left(\mathcal{U}_{\mathsf{C}(\nn)({\mu},{K})}\right)\geq& \frac{(N-1)^2}{2}h(\varepsilon)-O(N\ln{\mu})\\
\label{4.30}
+\frac{{K}}{2}\bigg((\nn-1)&h\left(\frac{M}{\nn-1}\right)-O(\ln\nn)\bigg)-O(1)\\
\label{4.4}
=&  \Omega\left(\nn^2 h(\varepsilon)+\nn^{1-3/{\mu}} {K}\right)
\end{align}
where $\varepsilon:=\phi/(N-1)\leq 1/2$. In (\ref{4.30}), we have used Stirling's approximation and the assumption that ${K}=o(\nn)$.
Continuing from (\ref{4.4}), since $2\nn h(\varepsilon)\geq {\mu}\ln(\nn/{\mu})$, for sufficiently large $\nn$,
\begin{align}
\label{4.5}
\mathbbm{H}\left(\mathcal{U}_{\mathsf{C}(\nn)({\mu},{K})}\right)= \Omega\left(\nn{\mu}\log\frac{\nn}{{\mu}}+\nn^{1-3/{\mu}} {K}\right).
\end{align}
Substituting (\ref{4.5}) into (\ref{4.20}), when $\nn\rightarrow\infty$, it must hold that $$\ns=\Omega\left({\mu}\log({\nn}/{{\mu}})+{K}/\nn^{3/{\mu}}\right)$$ to ensure that $\poe$ is smaller than $1/2$.
\end{proof}

\section{Proof of Theorem~\ref{thm:noisy_sample_complexity}}
\label{app:proof_noisy}
The structure of the proof is the same as Theorem~\ref{thm:4}. The converse follows directly by putting the bounds in (\ref{4.5}) and (\ref{eq:noisy_poe}) together. For proving the achievability, it is sufficient to show that with high probability (in $\nn$), $|Y_{i,j}-X_{i,j}|=o(1)$ for all $i,j\in\mathcal{V}$ where $X_{i,j}$ and $Y_{i,j}$ are the recovered and original $(i,j)$-th entry of the graph matrix. For the Gaussian IID ensemble considered, the $\ell_2$-norm of the inverse matrix $(\mathbf{B}^{\mathcal{K}}_{\overline{\mathcal{S}}})^{-1}$, equivalently, the minimal singular value of $\mathbf{B}^{\mathcal{K}}_{\overline{\mathcal{S}}}$ is strictly positive with probability $o(1)$ (see the proof of Lemma III-9 in~\cite{aeron2010information}). Using the Chernoff bound, with high probability,
\begin{align}
\label{eq:10.1}
    ||\mathbf{B}||^2_2\leq&||\mathbf{B}||_{\mathrm{F}}^2\leq C_1\nn\ns\sigma^2_{\mathrm{S}}, \\
    \label{eq:10.2}
    ||Z_j||^2_2\leq& C_2\nn\sigma^2_{\mathrm{N}}, \text{ for all } j\in\mathcal{V}
\end{align}
for some positive constants $C_1$ and $C_2$.
Noting that if $\nk\leq\mu$, then $\delta_{2\nk}<1$ with high probability, the bound in (\ref{eq:10.1}) and the bound on the $\ell_2$-norm of the inverse matrix $(\mathbf{B}^{\mathcal{K}}_{\overline{\mathcal{S}}})^{-1}$ imply $\eta=O(\nn^2\gamma)$, by applying our Theorem~\ref{thm:3}. Moreover, with Gaussian measurements, for each $\mu$-sparse vector $Y_j$ in $\mathbbm{R}^\nn$, $||X_j-Y_j||_2\leq C_3||Z_j||_2$ for some constant $C_3>0$ (\textit{cf.} Theorem 1 in~\cite{candes2006stable}) where $Y_j$ satisfies $\mathbf{B}Y_j+Z_j=A_j$ and $X_j$ is the optimal solution of (\ref{4.1})-(\ref{4.2}) (with $\mathbbm{F}\equiv\mathbbm{R}$). Therefore, $\Gamma=O(\gamma)$ and $\gamma=O(\sqrt{\nn}\sigma_{\mathrm{N}})$ using (\ref{eq:10.2}). Since $\eta=O(\nn^2\gamma)$, the condition $\sigma_{\mathrm{N}}=o(1/\nn^{5/2})$ guarantees that $\eta=o(1)$, whence $|Y_{i,j}-X_{i,j}|=o(1)$ for all $i,j\in\mathcal{V}$ and the proof is complete.






%


\end{document}